\documentclass[lettersize,journal]{IEEEtran}

\usepackage[papersize={8.5in,11in}, left=0.625in, right=0.625in, top=0.7in, bottom=1in]{geometry}
\usepackage{amssymb, amsmath, amsfonts, amsthm}
\usepackage{graphicx, cite, mathtools}
\usepackage{dsfont}
\usepackage[caption=false,font=normalsize]{subfig}

\newtheorem{theorem}{Theorem}
\newtheorem{lemma}{Lemma}
\newtheorem{remark}{Remark}

\usepackage{epstopdf}
\usepackage{url}
\usepackage{amsmath}
\def\mathclap#1{\text{\hbox to 0pt{\hss$\mathsurround=0pt#1$\hss}}}

\usepackage{tikz}

\allowdisplaybreaks
\usepackage{lipsum,graphicx,multicol}
\usepackage{algorithm}
\usepackage{algorithmic}

\begin{document}

\title{{An Evolutionary Game for Mobile User Access Mode Selection in sub-$6$ GHz/mmWave Cellular Networks}}

\author{Christodoulos Skouroumounis, \IEEEmembership{Member, IEEE}
	and~Ioannis Krikidis, \IEEEmembership{Fellow, IEEE}
	\thanks{Christodoulos Skouroumounis and Ioannis Krikidis are with the Department of Electrical and Computer Engineering, University of Cyprus, Nicosia 1678 (Email: {cskour03, krikidis}@ucy.ac.cy).}
	\thanks{This work was co-funded by the European Regional Development Fund and the Republic of Cyprus through the Research and Innovation Foundation under the project INFRASTRUCTURES/1216/0017 (IRIDA). It has also received funding from the European Research Council (ERC) under the European Union's Horizon 2020 research and innovation programme (Grant agreement No. 819819).}\vspace{-0.8cm}}

\maketitle
\begin{abstract}
By utilizing the combination of two powerful tools i.e., stochastic geometry (SG) and evolutionary game theory (EGT), in this paper, we study the problem of mobile user (MU) mode selection in heterogeneous sub-$6$ GHz/millimeter wave (mmWave) cellular networks. Particularly, by using SG tools, we first propose an analytical framework to assess the performance of the considered networks in terms of average signal-to-interference-plus-noise (SINR) ratio, average rate, and mobility-induced time overhead, for scenarios with user mobility{.} According to the SG-based framework, an EGT-based approach is presented to solve the problem of access mode selection. Specifically, two EGT-based models are considered, where for each MU its utility function depends on the average SINR and the average rate, respectively, while the time overhead is considered as a penalty term. A distributed algorithm is proposed to reach the evolutionary equilibrium, where the existence and stability of the equilibrium is theoretically analyzed and proved. Moreover, we extend the formulation by considering information delay exchange and evaluate its impact on the convergence of the proposed algorithm. Our results reveal that the proposed technique can offer better spectral efficiency and connectivity in heterogeneous sub-$6$ GHz/mmWave cellular networks with mobility, compared with the conventional access mode selection techniques.
\end{abstract}


\section{Introduction}
\IEEEPARstart{F}{uture} wireless networks, namely beyond fifth generation (B5G) and sixth generation (6G), are required to handle and accommodate a diverse set of {both static and mobile} end-user devices (e.g., remote sensors, unmanned aerial vehicles and autonomous cars), {designating the support of mobility as a fundamental aspect of wireless networks. Moreover, this unprecedented number of connected devices with such diverse requirements is also} contributing to the tremendous growing demand for network scalability, latency, and spectral efficiency (SE) \cite{AKY,ZHAN}. In order to meet this explosive throughput demand of future wireless connectivity, there has been an increasing interest in the synergy of network densification technique by using small cells (SCells) and the millimeter-wave (mmWave) communications \cite{AKY,ZHAN}. {Initially, the concept of network densification refers to the massive deployment of SCells (such as femptocells and picocells) by overlaying the conventional sub-$6$ GHz networks. Such heterogeneous network (HetNet) architectures can provide high throughput to the static users, but may significantly deteriorate the performance of mobile (i.e., moving) users (MUs). Indeed, the higher number of randomly deployed cells leads a MU to experience an increased number of handovers at cell boundaries, thereby resulting in potentially significant signaling overhead among the base stations (BSs) and MUs, compromising the HetNets performance [3]. \cite{XU}.}

In order to further enhance the network throughput, mmWave SCells have been considered as a promising technology for both the current and the future wireless networks, owing to the abundant spectrum resources in the mmWave band that can lead to multi-Gbps rates \cite{WAN}.  {Nevertheless, in comparison with the current sub-$6$ GHz communications, communications at mmWave frequencies are challenging since the channel suffers from severe path loss, atmospheric absorption, and environmental obstructions \cite{AND}. Fortunately, the short wavelength of the mmWaves signals allows the deployment of massive antenna arrays at transceivers to enhance the array gain and combat the higher propagation losses of the mmWave signals. However,} the highly directional mmWave communications lead to a more frequent service interruption between a BS and a MU, due to the beam switching and the beam misalignment events, which degrade the network performance \cite{KAL}. {The above-mentioned problem is further intensified in scenarios with mobility, since even a slight beam misalignment or environmental changes, such as link blockage, device rotation, etc., can cause considerable signal drop \cite{OZK,LIU2}.} Therefore, efficient mobility and handover management is an inherent challenge that needs to be addressed in heterogeneous sub-$6$ GHz/mmWave cellular networks.

In a heterogeneous sub-$6$ GHz/mmWave cellular network, it is critical for a MU to select a proper access mode i.e., whether to communicate with a sub-$6$ GHz or a mmWave BS. On the one hand, by associating with a legacy sub-$6$ GHz macro-cell (MCell), a MU faces a significantly decreased amount of handover processes, resulting to a robust and continuous network connectivity, but at the cost of a reduced spectral efficiency. On the other hand, the enormous spectral efficiency achieved by the association of a MU with a mmWave SCell comes with a cost of frequently interrupted network connectivity, inducing severe time overhead that may jeopardize the network performance. Hence, the overall balance of the counter-posed effects introduced by the mobility of MUs on large-scale sub-$6$ GHz/mmWave cellular networks need to be addressed.

\textit{Related Works:} A promising solution to improve the robustness of future wireless networks, is that mmWave BSs will be overlaid on conventional sub-$6$ GHz networks, where the sub-$6$ GHz BSs provide universal coverage, while the mmWave BSs provide high data rates in their range. Several research efforts have been carried out to evaluate the performance of sub-$6$ GHz/mmWave networks in the context of large-scale HetNets. In \cite{ELS}, the authors studied a HetNet consisting of sub-$6$ GHz MCells and mmWave SCells, where both the signal-to-interference-plus-noise ratio (SINR) and the rate coverage probability were evaluated, under various cell association schemes. Following a similar thought, in \cite{GHA} a biasing based strategy has been proposed for load balancing across a multi-band (i.e., sub-$6$ GHz/mmWave) HetNet. In \cite{SHI2}, the authors investigated the effect of the downlink/uplink decoupled association scheme in the context of heterogeneous sub-$6$ GHz/mmwave cellular networks, and analytical expressions for the rate, outage probability, and area throughput were derived. By considering different propagation characteristics of different mmWave frequency bands, the joint user association and resource allocation problem of a multi-band HetNet has been investigated in \cite{LIU}. In \cite{ZHA2}, the energy efficiency of multi-band HetNets with energy harvesting design has been investigated through an iterative gradient joint user association and power allocation algorithm. Aiming the ubiquitous connectivity in such networks, the authors in \cite{SKO} proposed a hybrid BS cooperation scheme, and the meta-distribution of the SINR was evaluated. 

The above-mentioned studies focus on the network performance for only the static case without considering the mobility scenarios. Although the integration of mmWave SCells into the existing sub-$6$ GHz networks, realized via heterogeneity and BS densification, tends to meet the desired spectral efficiency, the capacity gains are achieved at the expense of increased handover rates, and hence, higher service delays. The concept of handover in such networks is studied in \cite{POL}, where the authors proposed a dual connectivity (DC) framework that enables MU devices to communicate with both sub-$6$ GHz and mmWave BSs simultaneously. In \cite{KIB}, the authors have shown that the performance of a MU improves by employing the DC strategy when comparing with the traditional single connectivity (SC) in terms of the coverage probability. The impact of MUs' mobility on user-centric mmWave communications with multi-connectivity is evaluated in \cite{CHO}, where a state machine is developed to describe the mobility of the MUs. In \cite{MON}, the authors analyzed how to improve the quality-of-service in Long-Term Evolution-New Radio multi-band scenarios by comparing different channel measurement metrics. In \cite{ZHA}, each MU is allowed to have DC and the particle swarm optimization is adopted to maximize the average successful delivery probability. Nevertheless, the strict synchronization requirement for the DC approach, escalates the hardware complexity and information exchange delay \cite{POL}. The concept of SC, on the other hand, is a promising low-complexity and low-latency alternative approach that has been overlooked.

The proper access mode selection of a MU that employs the SC approach i.e., whether to communicate with a sub-$6$ GHz or a mmWave BS, in the context of heterogeneous sub-$6$ GHz/mmWave cellular networks is of paramount importance. Conventional approaches study the problem of MUs' access mode selection based on either the received SINR or the Euclidean distance from the serving BS \cite{YAN1}, which may not be suitable for networks with mobility. Evolutionary game theory (EGT) is an alternative and suitable approach to address the aforementioned problem. Specifically, EGT models the decision-making process of a population (i.e., group) of players, where each player evolves over time by gradually adjusting its action so that the payoff is maximized. In \cite{NIY}, the authors studied the problem of network selection in HetNets, and two algorithms were proposed based on population evolution and reinforcement-learning to solve the investigated problem. The authors in \cite{YAN2} proposed a dynamic network selection algorithm based on EGT for a fog-radio access network, and showed that the proposed algorithm has a better payoff than the conventional max rate-based scheme. In \cite{SEM}, the authors used a stochastic geometry (SG)-based approach to analyse the equilibrium's stability of the considered evolutionary game in two-tier HetNets, proposing an enhanced proportional fairness scheduler to improve the efficiency of resource utilisation. However, the above studies only focus on the achieved performance of a static user, neglecting users' mobility within the network.

\textit{Contributions:} Motivated by the above, in this paper, we study the dynamics of MU access mode selection for heterogeneous sub-$6$ GHz/mmWave cellular networks, consisting of sub-$6$ GHz and mmWave BSs. The main contributions of this paper can be summarized as follows.
\begin{itemize}
	\item We propose an analytical framework based on SG, which comprises the co-design of HetNets and mmWave SCells. The developed framework takes into account the ability of MUs to operate either at the sub-$6$ GHz or the mmWave frequency bands. Based on the developed framework, the heterogeneous sub-$6$ GHz/mmWave cellular network performance is assessed in terms of the attainable average SINR, the achievable average rate, and the required mobility-induced time overhead. The accuracy of the analytical results is validated by numerical studies.
	\item Based on the results obtained from the SG analysis, an EGT-based algorithm is proposed to solve the problem of MU access mode selection in the considered network deployments with respect to the mobility of the MUs, in order to achieve a reliable network connectivity. More specifically, we model the strategy adaptation process of the MUs by the replicator dynamics and the associated evolutionary equilibrium stability of the formulated evolutionary game has been analyzed. We finally extend the above formulation by considering information exchange delay and study its impact on the convergence of the proposed algorithm.
	\item Analytical expressions for the average achievable SINR, the average achievable rate, and the required mobility-induced time overhead are derived under the different MU access mode selections. Moreover, under specific practical assumptions, closed-form expressions for the Laplace transform of the received interference are derived. These closed-form expressions provide a quick and convenient methodology of evaluating the system's evolutionary equilibrium and obtaining insights into how key parameters affect the performance. {Our results reveal that the proposed EGT-based algorithm is an efficient tool to ensure reliable connectivity. Finally, compared to the conventional max rate-based access mode selection technique, the proposed technique significantly enhances the achieved average spectral efficiency of the considered deployment.}
\end{itemize}

\textit{Paper Organization:} Section \ref{SystemModel} introduces the network model together with the channel, blockage, mobility, and sectorized antenna models. The access mode selection strategy of MUs is formulated as an evolutionary game in Section \ref{EGT}. Section \ref{SG} provides the SG-based analysis of the average SINR, the average achieved rate, and the mobility-induced time overhead used in the utility functions of the game formulation. The evolutionary equilibrium and analysis of the stability of the equilibrium are presented in Section \ref{Equilibrium}. Simulation results are presented in Section \ref{Numerical}, followed by our conclusions in Section \ref{Conclusion}.

\section{System model}\label{SystemModel}
In this section, we provide details of the considered system model; the main mathematical notation related to the system model is summarized in Table \ref{Table1}.

\begin{table*}[t]\centering
	\caption{Summary of Notations.}\label{Table1}
	\scalebox{0.85}{
		\begin{tabular}{| l | l || l | l |}\hline
			\textbf{Notation} & \textbf{Description} & \textbf{Notation} & \textbf{Description}\\\hline
			$K$ & Total number of network tiers & $P_k$ & Transmit power of BSs in the $k$-th tier \\\hline
			$\Phi_k,\lambda_k$ & PPP of BSs in $k$-th tier of density $\lambda_k$ & $L(X,Y),a$ & Path-loss model and exponent \\\hline
			$\Phi_u,\lambda$ & PPP of MUs of density $\lambda$ & $\mu_k$ & Nakagami parameter for the $k$-th tier\\\hline
			$f_V(v)$ & MUs' velocity pdf & $\sigma^2$& White Gaussian noise\\\hline
			$\lambda_b$ & Spatial density of the blockages & $M_k,m_k,\phi_k$ & Sectorised antenna parameters for the $k$-th tier \\\hline
			$f_L(\cdot),f_W(\cdot)$ &  Blockages' lengths and widths pdfs& $\Phi_k^a,\lambda_k^a(r)$ & PPP and density of active BSs in the $k$-th tier \\\hline
			$\mathbb{E}[L], \mathbb{E}[W]$ & Mean length and width of blockages& $\mathcal{A}=\{\alpha_1,\dots,\alpha_K\}$ & Action set of each player \\\hline
			$\beta,p$ & Blockage parameters & $\Omega,\omega_\alpha$ & Total number of MUs and the MUs selecting action $\alpha$ \\\hline
			$p_{\rm L}(r)$ & LoS probability& $\chi_\alpha$ & Population share of action $\alpha$ \\\hline
			$f_k(r)$ & LoS Distance pdf & $w_1,w_2$ & EGT-based game weights \\\hline
			$T_m$ &  Predefined threshold for maximum time overhead& $\mathcal{I}_k,\mathcal{L}_{\mathcal{I}_k}(\cdot)$ & Interference and Laplace transform of the interference function \\\hline
			
	\end{tabular}}\vspace{-3mm}
\end{table*}
\subsection{Network Model}
The network is studied from a macroscopic point-of-view using SG. We consider an  orthogonal frequency division multiple access $K$-tier HetNet composed of a single sub-6 GHz MCell overlaid with $K-1$ mmWave SCells. The BSs belonging to the $k$-th tier, where $k=\{1,\cdots,K\}$, are modeled as a homogeneous Poisson point process (PPP) $\Phi_k=\{x_{i,k}\in\mathbb{R}^2,i\in\mathbb{N}^+\}$ with a spatial density $\lambda_k$ ${\rm BS/km}^2$. Moreover, the initial locations of the MUs follow an arbitrary independent point process $\Phi_u=\{u_i\in\mathbb{R}^2,i\in\mathbb{N}^+\}$ with spatial density $\lambda$ ${\rm MU/km}^2$. {Since multiple MUs can exist in the coverage area of a BS, a round-robin scheduling mechanism is employed, which randomly and without any priority schedules a single MU at each time slot to communicate with the assigned BS.} It is also assumed that the MUs are capable of operating at both sub-$6$ GHz and mmWave frequency bands, but each MU is solely served by a single network tier\footnote{{Although the cooperative communication of a MU with BSs from multiple tiers can enhance the network performance, the strict synchronization requirement designates such techniques as unsuitable for scenarios with strict time-delay constraints.}}. Fig. \ref{TopologyFig} illustrates a realization of a two-tier sub-$6$ GHz/mmWave cellular network i.e, $K = 2$, consisting of sub-$6$ GHz macro-cell and mmWave SCell BSs, where all MUs can be either static or mobile. {Without loss of generality and by following Slivnyak's theorem \cite{HAEb}, we execute the analysis for the typical MU, which is initially located at the origin, but the results hold for all MUs of the network.}

\begin{figure}[!t]
	\centering
	\includegraphics[width=0.6\linewidth]{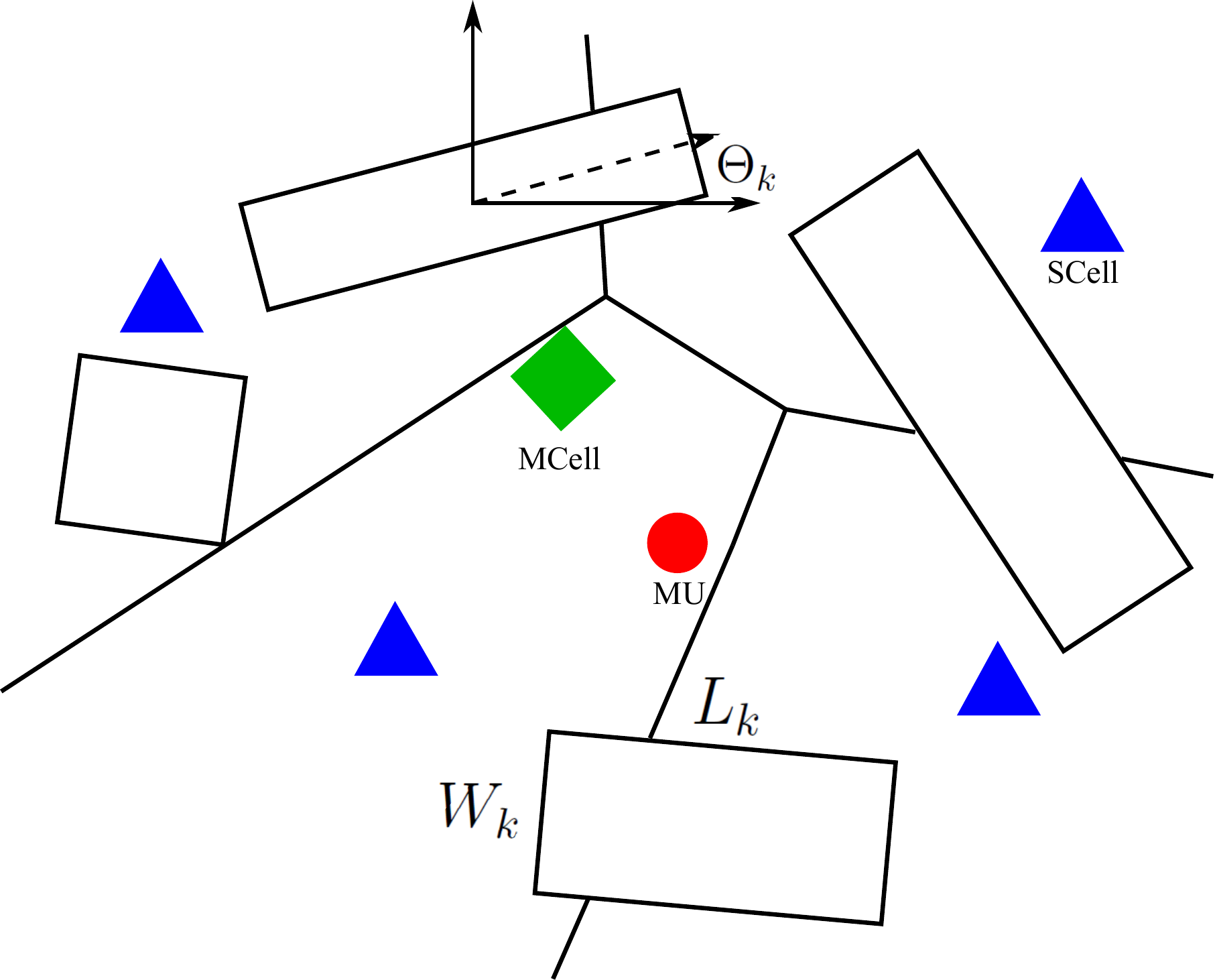}
	\caption{Network topology of a two-tier sub-$6$ GHz/mmWave cellular network, which consists of a sub-$6$ GHz macrocell (MCell) tier and a mmWave SCell tier, where all MUs can be either static or mobile. Blockages are modeled as a random process of rectangles, where the lengths $L_k$ and widths $W_k$ of the rectangles are assumed to be i.i.d. distributed, and the orientation $\Theta_k$ is assumed to be uniformly distributed in $(0, 2\pi]$.}
	\label{TopologyFig}\vspace{-0.5cm}
\end{figure}

We consider that the mobility of a MU within the network area is according to the random waypoint (RWP) model which can be described as follows. {At the beginning, all MUs are initially placed at uniformly random locations drawn from a uniform distribution $\Phi_u\in\mathbb{R}^2$. Thereafter, each MU moves with constant velocity towards a selected destination point (also known as waypoint) within the network area. In particular, motivated by the 3GPP mobility model \cite{POL}, we assume that each MU moves along a straight line with uniformly random speed $v$ (i.e., $v\sim\mathcal{U}(v_{\rm min},v_{\rm max})$) and direction $\phi$ (i.e., $\phi\sim\mathcal{U}(0,2\pi)$), independently of the other MUs. These simple and mathematically tractable random straight-line mobility models are widely adopted in the literature and can be regarded as benchmarks for evaluating more sophisticated models \cite{BAN}.} It is important to mention here that, each MU moves independently of other MUs, and thus, all nodes have the same stochastic mobility properties.
\subsection{Blockage and channel model}
A mmWave link can be either line-of-sight (LoS) or non-LoS (NLoS), depending on whether the BS is visible to the MU or not. More specifically, a transmitter is considered LoS by a receiver, if and only if their communication link is unobstructed by blockages. We assume that the blockages form a Boolean scheme of rectangles. More specifically, the centers of the rectangles (i.e., blockages) form a homogeneous PPP of density $\lambda_b$, while the lengths $L_k$ and widths $W_k$ of the rectangles are assumed to be i.i.d. distributed according to some pdf $f_L(x)$ and $f_W(x)$, respectively. The orientation of the rectangles i.e., $\Theta_k$, is assumed to be uniformly distributed in $(0, 2\pi]$. Blockage deployment under the aforementioned approach can be depicted in Fig. \ref{TopologyFig}. Under the adopted blockage model, the LoS probability function $p_{\rm L}(r)$, which depicts the probability that a link of length $r$ is LoS, is equal to $p_{\rm L}(r) = \exp(-\beta r-p)$, where $\beta=2\lambda_b(\mathbb{E}[L]+\mathbb{E}[W])/\pi$ and $p=\mathbb{E}[L]\mathbb{E}[W]$ \cite{AND}. {In this paper, the effect of the NLoS signals is ignored, and we focus on the analysis where the typical MU is only associated with a LoS transmitter and the interference stems from only LoS interferers \cite{AND}.}

We assume that all BSs that belong to the $k$-tier transmit with the same power $P_k$ (dBm), where $P_i>P_j$ if $i<j$. All channels in the network are assumed to experience both large-scale path-loss effects and small-scale fading. In particular, we model the large-scale path-loss between the transmitter located at $X$ and the receiver located at $Y$ by an unbounded singular path-loss model, $L(X,Y)=\|X-Y\|^a$, where $a > 2$ denotes the path-loss exponent. As for the small-scale fading, we assume independent Nakagami-$\mu$ fading, where different links are assumed to be independent and identically distributed. This is motivated by the fact that the Nakagami-$\mu$ distribution is a generalized distribution which can model different fading environments, where Rayleigh and Rician distributions are its special cases \cite{AND}. Let $g_{i,k}$ be the small scale fading of the link between the $i$-th BS from the $k$-th tier and the typical MU. Hence, under the Nakagami fading assumption, the power of the channel fading $g_{i,k}$ is a normalized Gamma random variable with shape parameter $\mu_k$ and scale parameter $1/\mu_k$ i.e., $|g_{i,k}|^2=h_{i,k}\sim\Gamma\left(\mu_k,1/\mu_k\right)$. Furthermore, we assume all wireless links exhibit additive white Gaussian noise with zero mean and variance $\sigma^2$. {Finally, due to the high directionality of mmWave signals, the impact of the ground reflections (or other types of reflections) on the achieved network performance is ignored.}

\subsection{Sectorized antenna model}
In such multi-band deployments, sub-$6$ GHz BSs aspire to the ubiquitous coverage of the MUs, while mmWave BSs mainly focus on providing high capacity to individual MUs. Motivated by this, we assume the employment of omni-directional antennas for all sub-$6$ GHz BSs, while all mmWave BSs are equipped with directional antennas. Furthermore, we assume that the MUs are equipped with an omni-directional antenna. Regarding the modeling the antenna directionality of the mmWave BSs, we adopt a sectorized antenna model that approximates the actual beam pattern with sufficient accuracy \cite{AND}. Specifically, the antenna array gain of the BSs that belongs in the $k$-th tier can be characterized by three values: $1$) the main-lobe beamwidth $\phi_k\in[0,2\pi]$, $2$) the main-lobe gain $M_k$ (dB), and $3$) the side-lobe gain $m_k$ (dB), where $M_k > m_k$. For simplicity, each BS that belongs in the $k$-th tier is assumed to have a codebook of $2^{n_k}$ possible beamforming vectors with $n_k\in\mathbb{N}$, where the patterns of these beamforming vectors have non-overlapping main-lobes designed to cover the full angular range i.e., $\phi_{k} = 2\pi/2^{n_k}=\pi/2^{n_k-1}$ \cite{KAL}. We assume that a perfect beam alignment can be achieved between each MU and its serving BS by using the estimated angles of arrival, resulting in an antenna array gain of $M_k$, denoted by $G_{0,k}$. Since the beams of all BSs are oriented towards their associated MUs, the direction of arrivals between interfering BSs and the typical MU is distributed uniformly in $[-\pi,\pi]$. Therefore, the antenna gain, $\mathcal{M}_k$, between the typical MU and a BS at $x\in\Phi_k$ is an i.i.d. discrete random variable described by
\begin{equation*}
\mathcal{M}_k = \begin{cases}
M_k, & \text{with probability } p_{M_k}= \frac{\phi_k}{2\pi}\\
m_k, & \text{with probability } p_{m_k}=  1-p_{M_k}.
\end{cases}
\end{equation*}
Note that, since all sub-$6$ GHz BSs (i.e., $k=1$) are equipped with omni-directional antennas, the antenna gain between a MU and a sub-$6$ GHz BS is equal to $\mathcal{M}_1=0$ dB. 

\section{Evolutionary Game Approach for MU Access Mode Selection}\label{EGT}
In this section, the evolutionary game formulation for the MU access mode selection is formulated, and the replicator dynamics is exploited in order to model the strategy adaptation process. Finally, the formulation is further extended by considering information exchange delay and its impact on the convergence of the developed algorithm is investigated.
\subsection{Game formulation}
The adaptive access mode selection among the network tiers of the HetNet deployment can be formulated as an evolutionary game as follows:
\begin{itemize}
	\item \textbf{Set of players:} In the considered MUs' access mode selection game, the set of MUs $\Phi_u$ denotes the set of players.
	\item \textbf{Set of actions:} The MUs (i.e., the players) are interested in selecting a suitable access mode. According to the system model, each player has $K$ access mode strategies. Accordingly, we define the action set for each player as $\mathcal{A}=\{\alpha_1,\cdots,\alpha_K\}$, which includes all possible access mode strategies.
	\item \textbf{Population:} In the context of evolutionary game, the set of players also constitutes the population. Let $\Omega$ and $\omega_\alpha$ denote the total number of MUs and the number of MUs selecting action $\alpha\in\mathcal{A}$, respectively. Then, the population share of action $\alpha$, is given by $\chi_\alpha=\frac{\omega_\alpha}{\Omega}$. It is important to mention here that the population share, $\chi_\alpha$, indicates the fraction of BSs within the network tier $\alpha$ that are active (e.g., the fraction of BSs that are selected for serving a set of MUs).
	\item \textbf{Payoff function:} The payoff function quantifies the performance satisfaction of MUs, for which two components are taken into account. The first component depicts the utility function which is associated with the SINR observed by a player when a certain access mode is selected. The second component is a penalty (or reward) term depending on whether the required time overhead incurred by the handover processes of the MUs exceeds the pre-defined threshold $T_{\rm m}$ or not. Specifically, the payoff of a MU with a velocity $u$ and an access mode strategy $\alpha$, is defined as
	\begin{equation}\label{Payoff}
	\pi_{\alpha}=w_1\mathcal{U}\left({\rm SINR}_\alpha\right)-w_2\left(T_{\rm HO}(\alpha,u)-T_{\rm m} \right),
	\end{equation}
	where $w_1$ and $w_2$ are biasing factors which can be determined according to which network tier (i.e., sub-$6$ GHz or mmWave) should be given priority in the MUs' allocation, and $\mathcal{U}\left({\rm SINR}_\alpha\right)$ denotes the utility function measuring the achieved performance. The penalty (or reward) term, $T_{\rm HO}(\alpha,u)$ is the required time overhead incurred by the handover processes of the MUs (see Section \ref{SG}), and $T_{\rm m}$ is a predefined threshold for the maximum time overhead required for the handover processes. A MU may receive a penalty or reward depending on whether the time overhead constraint is violated or not. The penalty or reward is modeled in the payoff function which will be explained in Section \ref{EGT}.
\end{itemize}
{As for the utility function of the investigated EGT, two well-investigated functions are considered, namely the average SINR and the average achievable rate, leading to two low-complexity MU access mode selection strategies.} Specifically, the considered utility functions are mathematically expressed as follows
\begin{equation}
\mathcal{U}^{(1)}({\rm SINR}_\alpha)=\mathbb{E}[{\rm SINR}_\alpha],
\end{equation}
and
\begin{equation}
\mathcal{U}^{(2)}({\rm SINR}_\alpha)=\mathbb{E}[\mathcal{R}_\alpha],
\end{equation}
where $\mathcal{R}_\alpha$ represents the achievable rate of a MU that selects the access mode strategy $\alpha$ i.e., $\mathcal{R}_\alpha=B_\alpha\log_2(1+{\rm SINR}_\alpha)$ (bps), and $B_\alpha$ depicts the allocated bandwidth of the channel for the $\alpha$-th tier. Then, for a MU selecting action $\alpha$, the corresponding payoff functions for each of the above utilities, $\pi_\alpha^{(1)}$ and $\pi_\alpha^{(2)}$, can be written as follows
\begin{equation}\label{Payoff1}
\pi_{\alpha}^{(1)}=w_1\mathbb{E}[{\rm SINR}_\alpha]-w_2\left(T_{\rm HO}(\alpha,u)-T_{\rm m} \right),
\end{equation} 
and
\begin{equation}\label{Payoff2}
\pi_{\alpha}^{(2)}=w_1\mathbb{E}[\mathcal{R}_\alpha]-w_2\left(T_{\rm HO}(\alpha,u)-T_{\rm m} \right).
\end{equation} 
Based on the above-mentioned definitions, two evolutionary games are defined, namely $\mathcal{G}^{(1)}$ and $\mathcal{G}^{(2)}$. Henceforth, the superscripts $``(1)"$ and $``(2)"$ in the payoff, the utility, and the population share depict the corresponding game. 
\subsection{Replicator dynamics}
EGT combines game theory with a dynamic evolutionary process, focusing on the dynamics of the strategy adaptation in the population, as opposed to traditional game theory, which emphasizes on the static equilibrium \cite{SEM}. According to the EGT, all players constantly modifying their behaviors, strategies, etc., and the successful strategies are adopted to obtain a better payoff. This is referred to as the evolution of the game $\mathcal{G}^{(\kappa)}$, where $\kappa\in\{1,2\}$, during which the strategy adaptation of the MUs will change the population share, $\chi_\alpha^{(\kappa)}$, over the time. Hence, the population share is a function of time $t$ which can be denoted as $\chi_\alpha^{(\kappa)}(t)$. In order to mathematically model and analyze the strategy adaptation process of the players in an evolutionary game, a set of ordinary differential equations is adopted, also known as replicator dynamics.

In the context of EGT, the evolutionary equilibrium is defined as the fixed points of the replicator dynamics; the acquisition of the evolutionary equilibrium of the game is equivalent to solving the set of ordinary differential equations given by the replicator dynamics. The evolutionary equilibrium is achieved through population evolution. For this purpose, all players initially adopt a randomly selected strategy $\alpha$. As the game is repeated, each player compares its own payoff with the average payoff of the entire population. If the obtained payoff is less than the average, in the next period, the player randomly selects another strategy. According to the replicator dynamics, the number of MUs selecting the access mode $\alpha$ will increase if the corresponding payoff is higher than the average (i.e., $\overline{\pi}^{(\kappa)}(t)>\pi_\alpha^{(\kappa)}(t)$), which is defined as follows
\begin{equation}
\dot{\chi}_\alpha^{(\kappa)}(t) = \varrho \chi_\alpha^{(\kappa)}(t)\left(\pi_\alpha^{(\kappa)}(t)-\overline{\pi}^{(\kappa)}(t)\right)\ \forall \alpha\in\mathcal{A},
\end{equation}
where $\kappa\in\{1,2\}$, $\varrho>0$ represents the parameter that control the speed of the MUs in observing and adjusting their access mode selection, and $\overline{\pi}_\alpha(t)$ is the average payoff of all MUs, which can be computed from
\begin{equation}
\overline{\pi}^{(\kappa)}(t)=\sum\nolimits_{\alpha\in\mathcal{A}}\pi_\alpha^{(\kappa)}(t) \chi_\alpha^{(\kappa)}(t).
\end{equation}
Based on the aforementioned discussion on the replicator dynamics of the evolutionary game $\mathcal{G}^{(\kappa)}$, where $\kappa\in\{1,2\}$, the population evolution can be described as Algorithm \ref{Algorithm}.

\begin{algorithm}[t!]
	\caption{Evolution algorithm for game $\mathcal{G}^{(\kappa)}$.}
	\begin{algorithmic}[1]\label{Algorithm}
		\small\STATE \textbf{Initialize}  The MU randomly selects an access mode.\\
		\STATE \textbf{Step 1} Each MU communicates with the selected access mode and observes the received payoff, which is calculated based on \eqref{Payoff}. The payoff and the selected access mode information are then sent to a central controller.\\
		\STATE \textbf{Step 2} The central controller calculates both the average payoff of the population, $\overline{\pi}^{(\kappa)}(t)$, and the population state, $\chi_\alpha^{(\kappa)}(t)$, and broadcasts it to all MUs.\\
		\STATE \textbf{Step 3} Each MU compares its own payoff with the average payoff of the population. For a MU selecting access mode $k$, if its payoff is less than the average payoff it would randomly switch to another access mode $j$, where $\{k,j\}=\{1\dots,K\}$, $j\neq k$, and $\pi_{\alpha_j}^{(\kappa)}(t)>\pi_{\alpha_k}^{(\kappa)}(t)$.\\
		\STATE Repeat from Step 1 to Step 3 until convergence.
	\end{algorithmic}
\end{algorithm}

\subsection{Delay in replicator dynamics}
As mentioned before, each MU calculates and shares its potential payoff corresponding to a certain strategy to the centralized controller. However, the latency observed in practical communication links leads to a delayed information exchange, affecting the evolution process of the population. Hence, the MUs may not make use of the latest population information when making their association decisions. For instance, a MU's access decision at time $t$ may be determined by the population shares information at time $t-\tau$ (time gap for $\tau$ units of time). Consequently, the delayed replicator dynamics for the access mode strategy selection can be modified as follows
\begin{equation*}
\dot{\chi}_\alpha^{(\kappa)}(t) = \chi_\alpha^{(\kappa)}(t-\tau)\left(\pi_\alpha^{(\kappa)}(t-\tau)-\overline{\pi}^{(\kappa)}(t)\right),
\end{equation*}
where $\kappa\in\{1,2\}$. By leveraging the delayed replicator dynamics, the impact of information exchange delays on the convergence of the strategy adaptation process can be evaluated. To quantify the player's payoff functions in the considered game formulation, in the following section, we use SG tools to obtain the average SINR, the average achievable rate and the mobility-induced time overhead for each access mode selection.
\section{Heterogeneous sub-$6$ GHz/mmWave Mobile Networks:\\ A Macroscopic point-of-view}\label{SG}
In this section, we analyze the players' payoff function in the considered game formulation presented above (in Section \ref{EGT}), in the context of SG. We first evaluate the statistic properties of the aggregate interference, where analytical and closed-form expressions for the Laplace transform of the interference function are derived. Thereafter, we examine the players' utility functions (i.e., $\mathcal{U}^{(1)}(\cdot)$ and $\mathcal{U}^{(2)}(\cdot)$) depending on the selected access mode strategy $\alpha$. Finally, by analyzing the time overhead required to complete the handover procedures that incurred by the MUs' mobility, the penalty/reward term of the players is assessed.
\subsection{Interference characterization}
Firstly, we investigate the received interference at the typical MU, where analytical and asymptotic expressions for the Laplace transform of the received interference are derived. The aggregate interference caused by the active BSs that belong in the $k$-th tier is denoted as $\mathcal{I}_k$. Specifically, the aggregate interference observed by the typical MU that is served by the BS at $x_{0,k}\in\Phi_k$ can be expressed as follows
\begin{equation}\label{Interference}
\mathcal{I}_k = \sum\nolimits_{x_{i,k}\in\Phi^{\rm a}_k\backslash x_{0,k}} \mathcal{M}_k P_k h_{i,k} \|x_{i,k}\|^{-a},
\end{equation}
where $\mathcal{M}_k = \{M_k,m_k\}$, $h_{i,k}$ is the channel fading of the typical MU with the BS at $x_{i,k}\in\Phi^{\rm a}_k\backslash x_{0,k}$, and $\Phi^{\rm a}_k$ is the point process that represents the active LoS BSs that belong in the $k$-th tier. The set of active LoS BSs i.e., $\Phi_k^{\rm a}$, can be modeled by a non-homogeneous PPP with density $\lambda_k^{\rm a}(r)$. Specifically, the intensity of the active LoS BSs from the $k$-th tier is given by 
\begin{equation}\label{ActiveInterfering}
	\lambda_k^{\rm a}(r) = \lambda_k\chi_kp_{\rm L}(r),
\end{equation}
where $\lambda_k$ is the intensity of the BSs that belong in the $k$-th tier, $p_{\rm L}(r)$ is the probability of a BS at distance $r$ to be LoS with the typical MU i.e., $p_{\rm L}(r) = \exp(-\beta r-p)$, and $\chi_k$ is the proportion of the BSs from the $k$-th tier that serve a MU. {It is important to mention here that, even if a MU is unable to communicate with any SCell mmWave BS (e.g. sparse network deployments), macro-cell sub-$6$ GHz BSs can still provide coverage to that particular MU.} Then, the cumulative distribution function (cdf) of the distance $R$ to the closest active LoS BS from the $k$-th tier is $\mathbb{P}[R>r]=\exp(-2\pi\lambda_k\chi_kU(r))$ \cite{AND}, and the pdf of the distance $R$, is given by
\begin{equation}\label{pdf}
	f_k(r)=2\pi\lambda_k\chi_kr\exp(-\beta r-2\pi\lambda_k\chi_kU(r)),
\end{equation}
where $U(r) = \frac{\exp(-p)}{\beta^2}\left(1-\exp(-\beta r)(1+\beta r)\right)$. Note that, for the sub-$6$ GHz MCell i.e., $k=1$, the density of the active LoS BSs is equal to $\lambda_1^{\rm L}(r) = \lambda_1\chi_1$. 
To facilitate the characterization of the aggregate interference, in the following Lemma, we compute the Laplace transform of the random variable $\mathcal{I}_k$ evaluated at $s$.
\begin{lemma}\label{LemmaLaplace}
	The Laplace transform of the aggregate interference function, $\mathcal{I}_k$, is given by
	\begin{align}
	\mathcal{L}_{\mathcal{I}_k}(s) &=\prod\nolimits_{\mathcal{M}_k}\exp\Bigg(-2\pi\lambda_k\chi_kp_{\mathcal{M}_k}\Omega_k(s)^2\nonumber\\&\quad\times\int_{\frac{r_0}{\Omega_k(s)}}^\infty\frac{z^{a+1}}{1+z^a}\exp\left(-\beta z\Omega_k(s)-p\right){\rm d}z\Bigg),
	\end{align}
	where $\Omega_k(s) = (s\mathcal{M}_kP_k)^\frac{1}{a}$ and $\mathcal{M}_k=\{M_k,m_k\}$.
\end{lemma}
\begin{proof}
	See Appendix \ref{Appendix2}.	
\end{proof}
Recall that $p_{\rm L}(r) = \exp\left[-\beta r-p\right]$ depicts the LoS probability function defined in Section \ref{SystemModel}, and captures the effect of building blockages. Although the expression in Lemma \ref{LemmaLaplace} can be evaluated by using numerical tools, this could be difficult due to the presence of multiple integrals. To address this, we further simplify the analysis by considering a step function for the blockage probability i.e., $p_{\rm L}(r)=\mathds{1}_{r<R_B}$ where $R_B$ is the maximum length of an LoS link (i.e., $R_B\approx \sqrt{2\exp[-p]}/\beta$ \cite{AND}). The validity of the aforementioned assumption will be shown in the numerical results. The Laplace transform of the aggregate interference function for $a = 4$, is evaluated at the following Remark.
\begin{remark}\label{Remark1}
	The Laplace transform of the aggregate interference function from the $k$-th tier, $\mathcal{I}_k$, with $p_{\rm L}(r)=\mathds{1}_{r<R_B}$ and $a=4$, is approximately equal
	{\small \begin{equation*}
	\widetilde{\mathcal{L}}_{\mathcal{I}_k}\!(s)\!\approx\!\prod\nolimits_{\mathcal{M}_k}\!\!\!\!\exp\!\left(\!-\pi\lambda_k\chi_kp_{\mathcal{M}_k}\Omega_k^2(s)\arctan\!\!\left[\!\frac{(R^2\!+\!r^2)\Omega_k^2(s)}{1-(Rr)^2}\right]\right),
	\end{equation*}}
	where $\Omega_k(s)=(s\mathcal{M}_kP_k)^\frac{1}{a}$.
\end{remark}
\begin{proof}
	By using \cite[3.194.1]{GRA}, the expression in Lemma \ref{LemmaLaplace} for the case with $p_{\rm L}(r)=\mathds{1}_{r<R_B}$ and $a=4$, can be simplified as follows
	\begin{align*}
	&\widetilde{\mathcal{L}}_{\mathcal{I}_k}(s)=\prod\nolimits_{G_k}\exp\Bigg(-\pi\lambda_k\chi_kp_{\mathcal{M}_k}\Omega_k^2(s)\\&\quad\times\left(\arctan\left[\left(\frac{R}{\Omega_k(s)}\right)^2\right]-\arctan\left[\left(\frac{r}{\Omega_k(s)}\right)^2\right]\right)\Bigg).
	\end{align*}
	Then, based on the identity $\arctan(x)-\arctan(y)=\arctan\left(\frac{x+y}{1-xy}\right)$, the final expression can be derived.
\end{proof}

Furthermore, in order to study the existence and the stability of the evolutionary equilibrium in Section \ref{Equilibrium}, it is necessary to derive a closed-form expression for the considered performance metrics. To address this, a tight upper bound for the Laplace transform of the aggregate interference function is evaluated at the following Remark.
\begin{remark}\label{Remark2}
	The Laplace transform of the aggregate interference function from the $k$-th tier, $\mathcal{I}_k$ is upper-bounded as follow
	\begin{equation*}
	\widetilde{\mathcal{L}}_{\mathcal{I}_k}(s)<\prod\nolimits_{\mathcal{M}_k}\exp\left(-\pi\lambda_kp_{\mathcal{M}_k}\chi_k\Omega_k^4(s)\left(\frac{1}{r^2}-\frac{1}{R^2}\right)\right).
	\end{equation*}
\end{remark}
\begin{proof}
	By employing the Taylor expansion of the inverse tangent function i.e., $\arctan(x)<\frac{\pi}{2}-\frac{1}{x}$ for $x\geq 1$, the final expression can be derived.
\end{proof}

We now proceed to the derivation of the expressions for the average SINR and the average achievable rate for a MU within the considered network deployment, depending on its selected access mode strategy.

\subsection{Average Signal-to-Interference-Plus-Noise Ratio}
Firstly, we investigate the distribution of the SINR observed at the typical MU that communicates with a BS from the $k$-th tier i.e., $\mathcal{P}_k(\vartheta)$, for which analytical and asymptotic expressions are derived. The SINR observed by the typical MU that is served by a BS at $x\in\Phi_k^{\rm a}$, denoted as $x_{0,k}$, can be written as follows
\begin{equation}\label{SINR}
	{\rm SINR}_k = \frac{P_kG_{0,k}h_{0,k}\|x_{0,k}\|^{-a}}{\sigma^2+\mathcal{I}_k},
\end{equation}
where $h_{0,k}$ represents the channel fading gain between the typical MU and its serving BS from the $k$-th tier, $\sigma^2$ is the variance of the additive white Gaussian noise at the receiver, and $\mathcal{I}_k$ represents the aggregate interference observed by a MU that is served by the BS at $x_{0,k}\in\Phi_k$ and is given by \eqref{Interference}. Based on \eqref{SINR}, the SINR distribution can be mathematically model as
\begin{align}
	\mathcal{P}_k(\vartheta)& =\mathbb{P}\left[{\rm SINR}_k>\vartheta\right]\nonumber\\ &=\mathbb{P}\left[\frac{P_kG_{0,k}h_{0,k}\|x_{0,k}\|^{-a}}{\sigma^2+\mathcal{I}_k}>\vartheta\right]\nonumber\\
	&=\mathbb{P}\left[h_{0,k}>\frac{r_0^a(\sigma^2+\mathcal{I}_k)\vartheta}{P_k G_{0,k}}\right].\label{AverageSINRProof}
\end{align}
To overcome the difficulty on Nakagami fading, Alzer's Lemma \cite{ALZ} on the complementary cdf of a gamma random variable with integer parameter can be applied. This relates the cdf of a gamma random variable into a weighted sum of the cdfs of exponential random variables. Hence, we can bound expression \eqref{AverageSINRProof} as
\begin{equation}
	\small\mathcal{P}_k(\vartheta)<\!\sum_{\xi=1}^{\mu_k}(-1)^{\xi+1}\!\binom{\mu_k}{\xi}\!\mathbb{E}_{r_0,\mathcal{I}_k}\!\left[\!\exp\!\left(\!-\frac{\eta\xi(\sigma_n^2\!+\!\mathcal{I}_k)\vartheta}{r_0^{-a}P_kG_{0,k}}\right)\!\right],\label{AverageSINRProof2}
\end{equation}
where $\eta = \mu_k(\mu_k!)^{-\frac{1}{\mu_k}}$. By using the expression for the Laplace transform of the interference function, $\mathcal{L}_{\mathcal{I}_k}(s)$, that is evaluated in Lemma \ref{LemmaLaplace}, \eqref{AverageSINRProof2} can be re-written as

{\small\begin{equation*}
\mathcal{P}_k(\vartheta)\!\!<\!\!\sum_{\xi=1}^{\mu_k}(-1)^{\xi+1}\!\binom{\mu_k}{\xi}\!\mathbb{E}_{r_0}\!\!\left[\!\mathcal{L}_{\mathcal{I}_k}\!\left(\!\frac{\eta \xi r_0^a \vartheta}{P_kG_{0,k}}\!\right)\!\exp\!\left(\!-\frac{\eta \xi r_0^a \sigma_n^2\vartheta}{P_kG_{0,k}}\!\right)\!\right].
	\end{equation*}}
Finally, by un-conditioning on $r_0$ the above expression with the pdf given by \eqref{pdf}, $\mathcal{P}_k(\vartheta)$ can be evaluated as
\begin{equation}\label{Coverage}
\small\mathcal{P}_k(\vartheta)\!<\!\sum_{\xi=1}^{\mu_k}(-1)^{\xi+1}\binom{\mu_k}{\xi}\!\int_{0}^\infty\!\mathcal{L}_{\mathcal{I}_k}(s)\exp\left(-s\sigma_n^2\right)f_k(r){\rm d}r,
\end{equation}
where $s = \frac{\eta \xi r^a \vartheta}{P_kG_{0,k}}$. To characterize the utility function for the evolutionary game $\mathcal{G}^{(1)}$ i.e., $\mathcal{U}_1({\rm SINR}_\alpha)$, in the following Theorem we compute the average SINR observed by the typical MU that is served by a BS from the $k$-th network tier.

\begin{theorem}
The average SINR observed by a MU which is served by a BS from the $k$-th tier, is given by
\begin{align}\label{MeanSINR}
		\mathbb{E}[{\rm SINR}_k]&<\sum_{\xi=1}^{\mu_k}(-1)^{\xi+1}\binom{\mu_k}{\xi}\nonumber\\&\times\!\int_{0}^\infty\!\left(\!\int_0^\infty\!\mathcal{L}_{\mathcal{I}_k}\!(s)\!\exp\!\left(\!-s\sigma_n^2\right)\!{\rm d}\vartheta\!\right)\!f_k(r){\rm d}r.
\end{align}
\end{theorem}
\begin{proof}
	The expectation of any positive random variable $X$ is given by $\mathbb{E}[X]=\int_0^\infty\mathbb{P}[X>x]{\rm d}x$. Hence, by unconditioning on $\vartheta$ the expression \eqref{Coverage}, the final expression can be obtained.
\end{proof}

{It is important to note here that, even though we are able to represent the Laplace transform of the interference in closed-form expression, it is still impossible to attain an exact closed-form expression for the coverage performance, and hence, the average achievable rate. Despite this fact, the derived expressions provide a quick and convenient methodology of evaluating the system's performance and obtaining insights into how key system parameters affect the performance.}
\subsection{Average achievable rate}
The average achievable rate represents the information rate that can be transmitted over a given bandwidth for the considered system model. More specifically, the average achievable rate of a MU that is served by the BS at $x_{0,k}\in\Phi_k$, can be mathematically described with the probability
\begin{align}
\mathcal{R}_k &= \mathbb{E}\left[B_k\log_2\left(1+{\rm SINR}_k\right)\right]\nonumber\\
&=\int_0^\infty\mathbb{P}\left[B_k\log_2\left(1+{\rm SINR}_k\right)>x\right]{\rm d}x\label{eq1}
\end{align}
where $B_k$ depicts the allocated bandwidth of the channel for the $k$-th tier, ${\rm SINR}_k$ is the observed SINR which is given by \eqref{SINR}, and \eqref{eq1} follows from the fact that the expectation of any positive random variable $X$ is given by $\mathbb{E}[X]=\int_0^\infty\mathbb{P}[X>x]{\rm d}x$. {By utilizing the distribution of the SINR and by the change of variable $\vartheta\rightarrow 2^\frac{x}{B_k}-1$, the above expression can be re-written as
\begin{equation}\label{Rate}
\mathcal{R}_k = \frac{B_k}{{\rm ln}(2)}\int_0^\infty\frac{\mathcal{P}_k(\vartheta)}{\vartheta+1}{\rm d}\vartheta,
\end{equation}
where $\mathcal{P}_k(\vartheta)$ represents the achieved coverage probability, that is given by \eqref{Coverage}.} The following theorem analytically derives the expression for the average achievable rate, $\mathcal{R}_k$, depending on the MU's access mode selection.

\begin{theorem}\label{Theorem1}
	The average achievable rate of a MU that is served by the $k$-th tier, is given by
	\begin{align}\label{MeanRate}
	\mathcal{R}_k < \frac{B_k}{{\rm ln}(2)}\sum_{\xi=1}^{\mu_k}\int_0^\infty&\Bigg(\int_0^\infty\frac{(-1)^{\xi+1}}{\vartheta+1}\binom{\mu_k}{\xi}\exp\left(-s\sigma_n^2\right)\nonumber\\&\qquad\qquad\times\mathcal{L}_{\mathcal{I}_k}(s)f_R(r){\rm d}r\Bigg){\rm d}\vartheta,
	\end{align}
	where $s = \frac{\eta \xi \vartheta r^a}{P_kG_{0,k}}$, $\eta = \mu(\mu!)^{-\frac{1}{\mu}}$, $f_R(r)$ denotes the pdf of the distance between a MU and its serving BS from the $k$-th tier, and $\mathcal{L}_{\mathcal{I}_k}$ is the Laplace transform of the interference function evaluated at $s$, which is given in Lemma \ref{LemmaLaplace}.
\end{theorem}

\begin{proof}
	By substituting the expression \eqref{Coverage} in \eqref{Rate}, we conclude to the desired expression.
\end{proof}

\setcounter{equation}{22}
\begin{figure*}[t!]{\small
		\begin{align}\label{Dynamics1}
		\dot{\chi}^{(1)}_\alpha(t)&=\varrho \chi^{(1)}_\alpha(t)\Bigg(w_1\sum_{\xi=1}^{\mu_k}(-1)^{\xi+1}\binom{\mu_k}{\xi}\int_{0}^\infty\left(\int_0^\infty\mathcal{L}_{\mathcal{I}_k}(s)\exp\left(-s\sigma_n^2\right){\rm d}\vartheta\right)f_k(r){\rm d}r\nonumber\\
		& +w_2\left(\frac{2^{n_k}uT_a\sqrt{\lambda_k}}{\pi}+\left(1+\frac{4u\sqrt{\lambda_k}}{\pi}-\exp\left(-\frac{u\lambda_b}{2\pi}\left(R\zeta+\mathbb{E}[W]{\rm Si}[2\pi]\right)\right)\right)T_s\right)\Bigg),
		\end{align}
		\begin{align}\label{Dynamics2}
		\dot{\chi}^{(2)}_\alpha(t)&=\varrho \chi^{(2)}_\alpha(t)\Bigg(\frac{w_1B_k}{{\rm ln}(2)}\sum_{\xi=1}^{\mu_k}\int_0^\infty\left(\int_0^\infty\frac{(-1)^{\xi+1}}{\vartheta+1}\binom{\mu_k}{\xi}\exp\left(-s\sigma_n^2\right)\mathcal{L}_{\mathcal{I}_k}(s)f_R(r){\rm d}r\right){\rm d}\vartheta\nonumber\\
		& +w_2\left(\frac{2^{n_k}uT_a\sqrt{\lambda_k}}{\pi}+\left(1+\frac{4u\sqrt{\lambda_k}}{\pi}-\exp\left(-\frac{u\lambda_b}{2\pi}\left(R\zeta+\mathbb{E}[W]{\rm Si}[2\pi]\right)\right)\right)T_s\right)\Bigg),
		\end{align}}
	\hrulefill\vspace{-0.5cm}
\end{figure*}
\setcounter{equation}{19}

\subsection{Handover rate}

\begin{figure}[t!]
	\centering\includegraphics[width=0.7\linewidth]{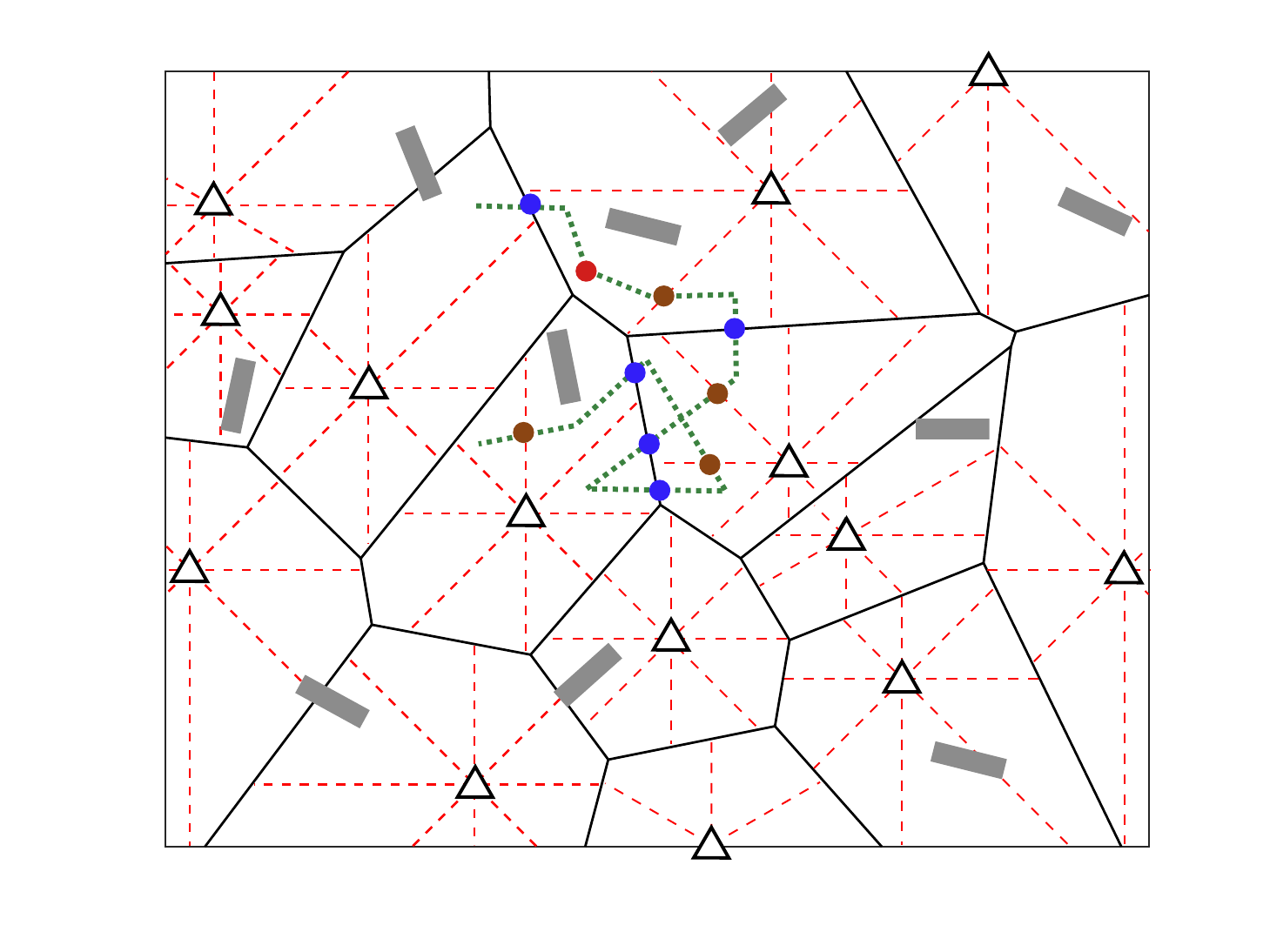}\caption{The Voronoi tessellation of a single-tier cellular network, where the BSs, the waypoints of a MU, and the blockages are represented by triangles, points, and rectangles, respectively. In the considered network deployment, each BS has $2^3 = 8$ beams i.e., $n_k = 3$, where the beam boundaries are depicted by red dashed lines. A MU's trajectory is illustrated by the dotted green line. Blue, brown, and red points represent the a-InterCH locations, the b-InterCH locations, and the IntraCH locations, respectively.
	}\label{Topology}\vspace{-0.5cm}
\end{figure}

In the depicted system model, Fig. \ref{Topology}, several BS handover processes may occur due to the MUs' mobility within the network area. Generally, a handover process is triggered when a MU re-selects a new beam from either its current serving BS, also known as \textit{intra-cell handover} (IntraCH), or a new serving BS, also known as \textit{inter-cell handover} (InterCH), aiming to stay connected with the best serving BS beam.

IntraCH procedure refers to the handover of a MU to a new beam from the current serving BS. Specifically, this procedure is executed when a MU crosses a beam boundary within the cell of its serving BS i.e., moves from the main lobe of one beam to that of another beam. Thus, a new beam has to be reselected at the beam boundary within the Voronoi cell of a BS, in order for the MU to stay connected with the main lobe of the best beam. In Fig. \ref{Topology}, the locations of beam reselections are represented by brown filled circles. Under the considered system model, the average rate of beam reselections for the $k$-th tier, $\delta_r(n)$, is equal to \cite{KAL}
\begin{equation}
\delta_{r}(n_k) = \frac{2^{n_k}\sqrt{\lambda_k}}{\pi}u.
\end{equation} 
 
As for the InterCH procedures, these can be classified into two categories, namely association-related InterCH (a-InterCH) and blockage-related InterCH (b-InterCH). Initially, based on the standard handover mechanism defined in the 3GPP specifications \cite{JAI}, the a-InterCH procedure refers to the BS handover that occurs when the received signal strength (RSS) of the candidate/target BS exceeds that of the serving BS. The main purpose of such mechanism is to to maintain the best connectivity throughout the MU's trajectory. In the depicted single-tier association case, an a-InterCH procedure is executed when a MU crosses a voronoi cell, by connecting the user to the BS in the center of that voronoi cell. Thus, a new beam at the next serving BS need to be selected during an a-InterCH procedure. Based on \cite{ARS}, the time intensity of cell boundary crossing i.e., BS handover, is $\delta_a=\frac{4\sqrt{\lambda_k}}{\pi}u$.

In addition to the aforementioned, the peculiarities of the mmWave propagation may call for much more frequent cell handovers. Indeed, the existence of blockages may cause frequent interruptions to the LoS path between a MU and its serving BS. This, in turn, leads to a rapid degradation in the RSS that could result in unwanted outages. To ensure the MUs' ubiquity and reliable connectivity in dense mmWave deployments, b-InterCH processes are executed, where a MU can handover to a new beam of other available BSs, if the current serving BS gets blocked. In the following Lemma, we analytically evaluate the time intensity of BS handover due to the existence of blockage effects.

\begin{lemma}\label{BlockageHandover}
	The average rate of BS handover due to blockage effects for the MU moving according to the RWP mode, is given by
	\begin{equation}
	\delta_b(u) = 1-\exp\left(-\frac{u\lambda_b}{2\pi}\left(R\zeta+\mathbb{E}[W]{\rm Si}[2\pi]\right)\right),
	\end{equation} 
	where $\zeta$ is a constant variable that is equal to $\zeta=\gamma+\log[\pi]-{\rm Ci}[\pi]$, $\gamma$ represents the Euler–Mascheroni constant; ${\rm Si}[\cdot]$ and ${\rm Ci}[\cdot]$ denote the sine and cosine integral functions, respectively.
\end{lemma}
\begin{proof}
	See Appendix \ref{Appendix3}.
\end{proof}

\subsection{Mobility-Induced Time Overhead}
Such handover processes result in a huge time overhead due to the time spent in beam sweeping and alignment. Specifically, the overall time required for successful execution of the handover processes consists of two components. Initially, the successful operation of both InterCH procedures (i.e., a-InterCH and b-InterCH) involves a beam sweeping process which requires time $T_s$. {It is important to note that, the handover events imposed by the considered InterCH processes are not independent. However, due to the tedious analysis for the exact number of InterCH events, in this paper, we assume that the a-InterCH and b-InterCH processes are independent between each other, which provides a tight lower bound for the performance. The validity of the assumptions will be shown in the numerical results.} Finally, the IntraCH procedures contribute to the overall time overhead with the time required for the beam alignment, $T_a$. By assuming a baseband processing chain for each antenna panel \cite{KAL}, the total average overhead per unit time of a MU that is served by a BS from the $k$-th tier, is given by
\begin{equation}\label{Time}
T_{\rm HO}(k,u)=\delta_r(k,u)T_a+\left(\delta_a(k,u)+\delta_b(u)\right)T_s.
\end{equation}
{It is important to mention here that, multiple handover procedures can be triggered during a MU's movement towards its destination. Under such scenario, the average rate of the considered handover procedures (i.e., IntraCH and/or InterCH processes) further increases, escalating the time overhead required for the incurred handover procedures.}
\setcounter{equation}{25}
\begin{figure*}[t!]{\small
		\begin{equation}\label{PF}
		\pi_\varrho^{(1)}(t)\!=\!w_1\!\prod_{\mathcal{M}_\varrho}\frac{p_{\mathcal{M}_\varrho}G_{0,\varrho}R^2}{\pi\lambda_\varrho\chi^{(1)}_\varrho\mathcal{M}_\varrho\widetilde{r}^2(R^2-\widetilde{r}^2)}\!+\!w_2\!\left(\frac{2^{n_\varrho}uT_a\sqrt{\lambda_\varrho}}{\pi}\!+\!\left(1\!+\!\frac{4u\sqrt{\lambda_\varrho}}{\pi}\!-\!\exp\left(-\frac{u\lambda_b}{2\pi}\left(\widetilde{r}\zeta+\mathbb{E}[W]{\rm Si}[2\pi]\right)\right)\right)T_s\right),
		\end{equation}}
	\hrulefill\vspace{-0.5cm}
\end{figure*}
\setcounter{equation}{24}
\section{Evolutionary Equilibrium and Stability Analysis}\label{Equilibrium}

\subsection{Existence of evolutionary equilibrium}
In the context of EGT, the evolutionary equilibrium refers to the fixed points of the replicator dynamics, which is considered to be the solution of the evolutionary games $\mathcal{G}^{(1)}$ and $\mathcal{G}^{(2)}$ for the access mode selection \cite{SEM}. More specifically, the population state at the equilibrium point does not change, and therefore, the rate of the strategy adaptation will be zero (i.e., $\dot{\chi}^{(1)}_\alpha(t)=0$ and $\dot{\chi}^{(2)}_\alpha(t)=0$). At the evolutionary equilibrium, all players (i.e., MUs) obtain the same payoff as the average payoff of the population, and hence, no player has the willingness to change its strategy. In this way, the evolutionary equilibrium can also provide fairness among the MUs.

By using the expressions given in \eqref{MeanSINR}, \eqref{MeanRate}, and \eqref{Time}, the replicator dynamics $\dot{\chi}^{(1)}_\alpha(t)$ and $\dot{\chi}^{(2)}_\alpha(t)$ of the investigated games $\mathcal{G}^{(1)}$ and $\mathcal{G}^{(2)}$, respectively, can be expressed by \eqref{Dynamics1} and \eqref{Dynamics2}, respectively, for all $\alpha\!\in\!\mathcal{A}$, where $s\! =\! \frac{\eta \xi \vartheta r^a}{P_kG_{0,k}}$, $\eta = \mu(\mu!)^{-\frac{1}{\mu}}$, and $\mathcal{L}_{\mathcal{I}_k}(s)$ depicts the Laplace transform of the aggregate interference function, that is given in Lemma \ref{LemmaLaplace}. 

As we mentioned earlier, the computation of the evolutionary equilibrium corresponds in solving the system of algebraic equations \eqref{Dynamics1} and \eqref{Dynamics2} by setting the left hand side of the replicator dynamics to zero, i.e., $\dot{\chi}^{(1)}_\alpha(t)=0$ and $\dot{\chi}^{(2)}_\alpha(t)=0$. It is important to mention here that, two types of evolutionary equilibrium can be investigated in the context of replicator dynamics, namely boundary and interior evolutionary equilibrium. On the one hand, the boundary evolutionary equilibrium of game $\mathcal{G}^{(\kappa)}$, where $\kappa=\{1,2\}$, refers to the scenario where there exists a population share $\chi_\alpha^{(\kappa)}=1$, while $\chi_b^{(\kappa)}=0$ for all $b\neq\alpha\in\mathcal{A}$. On the other hand, the interior equilibrium refers to scenario where $\chi_\alpha^{(\kappa)}\in(0,1),\ \forall\alpha\in\mathcal{A}$. It is obvious that boundary evolutionary equilibrium is not stable because any small perturbation will make the system deviate from the equilibrium state. Therefore, in the following section, we evaluate the stability of the interior evolutionary equilibrium.
\subsection{Stability of evolutionary equilibrium}
According to the definition of the replicator dynamics, the stability of the interior evolutionary equilibrium is evaluated by solving $\dot{\chi}^{(1)}_\alpha(t)=0$ and $\dot{\chi}^{(2)}_\alpha(t)=0$. In other words, the eigenvalues of the Jacobian matrix, which correspond to the replicator dynamics, need to be evaluated; then, the system is stable if all eigenvalues have a negative real part \cite{SEM}.

The stability of an evolutionary game is analytically tractable in cases of closed-form expressions for the replicator dynamics \eqref{Dynamics1} and \eqref{Dynamics2}, which are obtained from the SG-based analysis. In the following, we analyze the stability of the equilibrium of the investigated game $\mathcal{G}^{(1)}$, under specific practical assumptions. The stability of the considered evolutionary game $\mathcal{G}^{(2)}$ is not analyzed since $\dot{\chi}^{(2)}_\alpha(t)$ is not available in closed-form. Nevertheless, the stability of the game $\mathcal{G}^{(2)}$ will be obtained and illustrated by numerical simulation in Section \ref{Numerical}.

In order to derive closed-form expressions for the average SINR i.e., $\mathbb{E}[{\rm SINR}_k]$, we adopt the upper bound of the Laplace transform of the aggregate interference, which is evaluated in Remark \ref{Remark2}. We further assume that the small-scale fading between the typical MU and its associated BS is modeled as Rayleigh fading. Thus, we assume that the Nakagami-m parameters for the LoS communication links fulfil the expression $\mu_k=1,\ \forall k\in\{1,\dots,K\}$. Finally, the analysis for the stability of the interior evolutionary equilibrium of the game $\mathcal{G}^{(1)}$ is performed for the scenario of a two-tier HetNet deployment, where all MUs are at distance $\widetilde{r}$ from their serving BSs. In the following Remark, we provide a closed-form expression for the average SINR, $\mathbb{E}[{\rm SINR}_k]$, which will be useful for studying the existence and the stability of the evolutionary equilibrium.
\setcounter{equation}{24}
\begin{remark}\label{Simplified}
	The average SINR of a MU at distance $\widetilde{r}$ from its serving BS from the $k$-th tier with $p_{\rm L}(r)=\mathds{1}_{r<R_B}$, $a=4$, and Rayleigh fading, is given by
	\begin{equation}
		\mathbb{E}[{\rm SINR}_k]=\prod_{\mathcal{M}_k}\frac{p_{\mathcal{M}_k}G_{0,k}R^2}{\pi\lambda_k\chi^{(1)}_k\mathcal{M}_k\widetilde{r}^2(R^2-\widetilde{r}^2)}.
	\end{equation}
\end{remark}
\begin{proof}
	By setting $\mu_k=1\ \forall k\in\{1,\dots,K\}$ and by using \cite[3.310]{GRA}, the final expression can be derived.
\end{proof}

By substituting the upper bound of the average SINR, that is derived in Remark \ref{Simplified}, in \eqref{Dynamics1}, the payoff function of each strategy $\varrho$ is given by \eqref{PF}, where $\varrho=\{1,2\}$. For the stability of the system, we state the following Theorem, where we drop the notation $(t)$ for simplicity.
\setcounter{equation}{26}
\begin{theorem}\label{Stability}
	For a two network tier deployment and Rayleigh fading, the interior evolutionary equilibrium of the game $\mathcal{G}^{(1)}$ is asymptotically stable.
\end{theorem}
\begin{proof}
	See Appendix \ref{Appendix4}.
\end{proof}
It is important to mention here that, the above Theorem does not always holds for the delayed replicator dynamics. Based on the occurred delay on the strategy adaptation, the evolutionary equilibrium can be either stable or not. More specifically, if small delays occurred on the replicator dynamics, the evolutionary equilibrium is stable; otherwise, the evolutionary equilibrium is not stable. Due to the high non-linearity of the delayed replicator dynamics \cite{SEM}, the convergence of the proposed EGT-based access mode algorithm is numerically evaluated in Section \ref{Numerical}.

\section{Numerical Results}\label{Numerical}

\begin{table*}[t]\centering
	\caption{Simulation Parameters.}\label{Table2}
	\scalebox{0.85}{
		\begin{tabular}{| l | l || l | l |}\hline
			\textbf{Parameter} & \textbf{Value} & \textbf{Parameter} & \textbf{Value}\\\hline
			Number of Tiers ($K$) & $3$ & Side lobe gain ($m_1,m_2,m_3$) & $\{0,-5,-10\}$ dB \\\hline
			Radius ($\rho$) & $1000$ m & Beamwidth ($\phi_1,\phi_2,\phi_3$) & $\{2\pi,\pi/3,\pi/6\}$ \\\hline
			BSs' Density ($\lambda_1,\lambda_2,\lambda_3$) & $\{5,100,500\}$ BSs/km$^2$ & Weights ($w_1,w_2$) & $\{1,1\}$  \\\hline
			Transmit power ($P_1,P_2,P_3$) & $\{40,35,30\}$ dBm & Blockage density ($\lambda_b$) & $100$ Blockages/km$^2$ \\\hline
			MUs' density ($\lambda$) & $500$ MUs/km$^2$ & Mean blockage length ($\mathbb{E}[L]$) & $10$ m \\\hline
			Maximum speed ($v_{\rm max}$) & $100$ km/h & Mean blockage width ($\mathbb{E}[W]$) & $10$ m  \\\hline
			Nakagami parameters ($\mu_1,\mu_2,\mu_3$) & $\{1,4,4\}$ & Beam alignment overhead  ($T_a$) & $1$ ms  \\\hline
			Main lobe gain ($M_1,M_2,M_3$) & $\{0,5,10\}$ dB & Beam sweeping overhead ($T_s$) & $1$ ms  \\\hline
	\end{tabular}}\vspace{-3mm}
\end{table*}

{In this section, the performance of the proposed MU access mode selection strategy is evaluated in the context of three-tier HetNet composed of a single sub-$6$ GHz MCell (denoted as ``Tier $1$'') overlaid with $2$ mmWave SCells (denoted as ``Tier $2$'' and ``Tier $3$''). A summary of the model parameters is provided in Table \ref{Table2}. Please note that, the selection of the simulation parameters is made for the sake of the presentation. The use of different values leads to a shifted network performance, but with the same observations and conclusions.}

\begin{figure}[!t]
	\centering
	\subfloat[]{\includegraphics[width=0.595\linewidth]{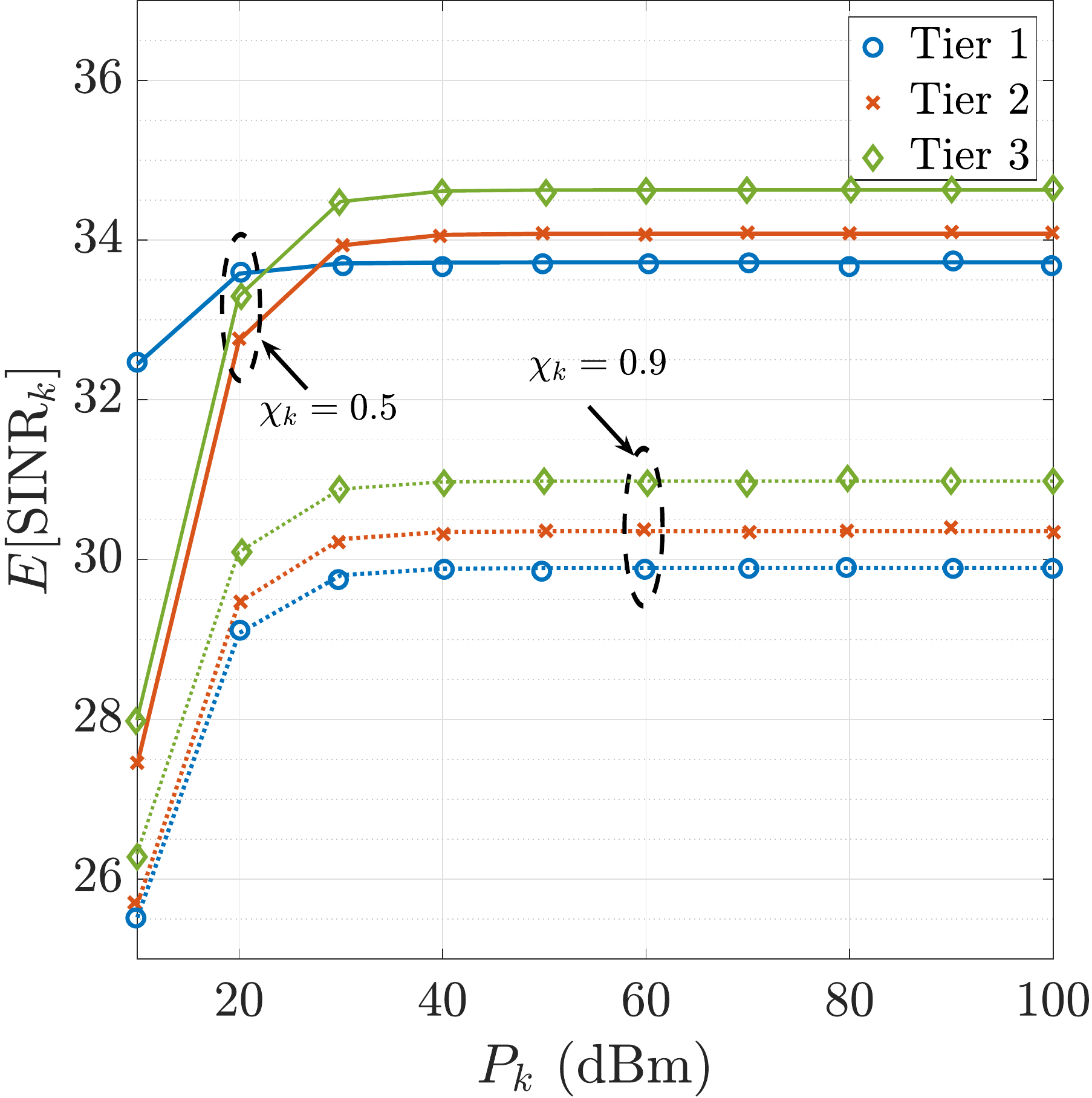}%
		\label{FirstFigure}}
	\hfil
	\subfloat[]{\includegraphics[width=0.6\linewidth]{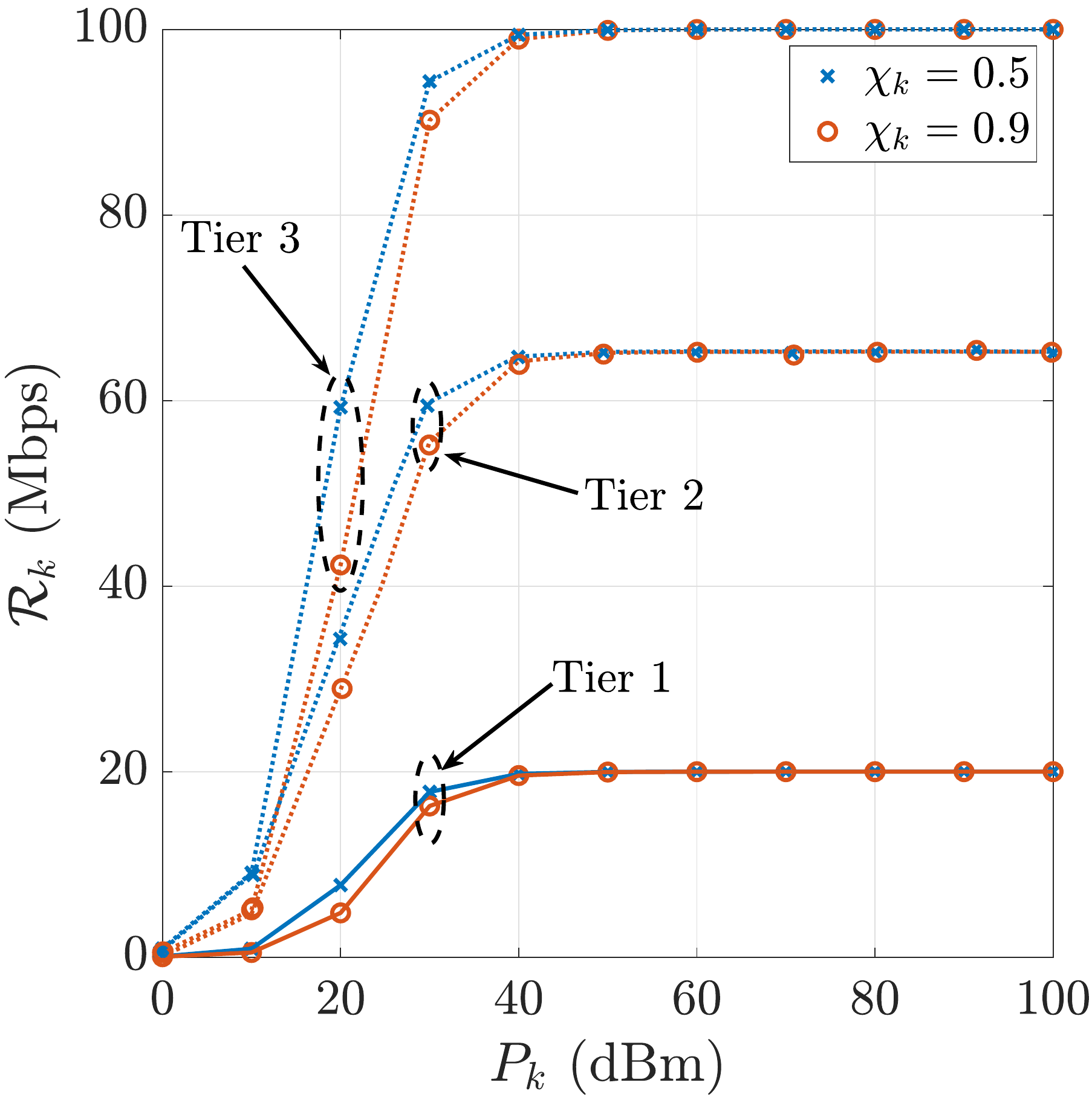}%
		\label{SecondFigure}}
	\caption{The payoff functions for the considered EGT-based games $\mathcal{G}^{(1)}$ and $\mathcal{G}^{(2)}$, where $\chi_k=\{0.5,0.9\}$.}\vspace{-0.5cm}
\end{figure}

Fig. \ref{FirstFigure} shows the achieved average SINR by all access mode selections of the considered three-tier HetNet deployments in terms of the BSs' transmission power $P_k$, where $\chi_k=\{0.5,0.9\}$. It is clear from the figure that the average SINR converges to a constant floor in all cases for high transmission powers. This is due to the fact that as the transmission power of the network's nodes increases, the noise in the network becomes negligible. The mmWave SCells (i.e., cases ``Tier $2$'' and ``Tier $3$'') achieve slightly better performance compared to the sub-$6$ GHz MCells (i.e., case ``Tier $1$'') since the unique features of mmWave signals assist in suppressing the multi-user interference. Nevertheless, the sub-$6$ GHz MCells suffer less in terms of path-loss attenuation and therefore it performs significantly better for low transmission powers, whereas the performance of the mmWave SCells is degraded by the severe propagation losses. Another important observation is that, the performance achieved by both sub-$6$ GHz MCells and mmWave SCells is adversely affected by the increase in the proportion of the served MUs. This was expected, since by increasing the number of MUs served by a certain network tier, more and more BSs are activated, leading to an increased multi-user interference, and therefore the average SINR is significantly decreases. The same behaviour is also observed in Fig. \ref{SecondFigure}, which depicts the average achievable rate with $\chi_k=\{0.5,0.9\}$ for all access mode selections. As expected, the performance of all network tiers improves with the increase of transmit power. However, it is obvious from Fig. \ref{SecondFigure} that the performance of MUs that are associated with the mmWave SCells (i.e., cases ``Tier $2$'' and ``Tier $3$'') is significantly higher than that of MUs that are associated with the sub-$6$ GHz MCells (i.e., case ``Tier $1$''). This is explained by the fact that the mmWave SCells offer an enormous channel bandwidth for their communication with the MUs, as opposed to the scarce channel resources of sub-$6$ GHz MCells. Finally, the agreement between the theoretical curves (solid and dashed lines) and the simulation results (markers) validates our mathematical analysis.

\begin{figure}
	\centering\includegraphics[width=0.6\linewidth]{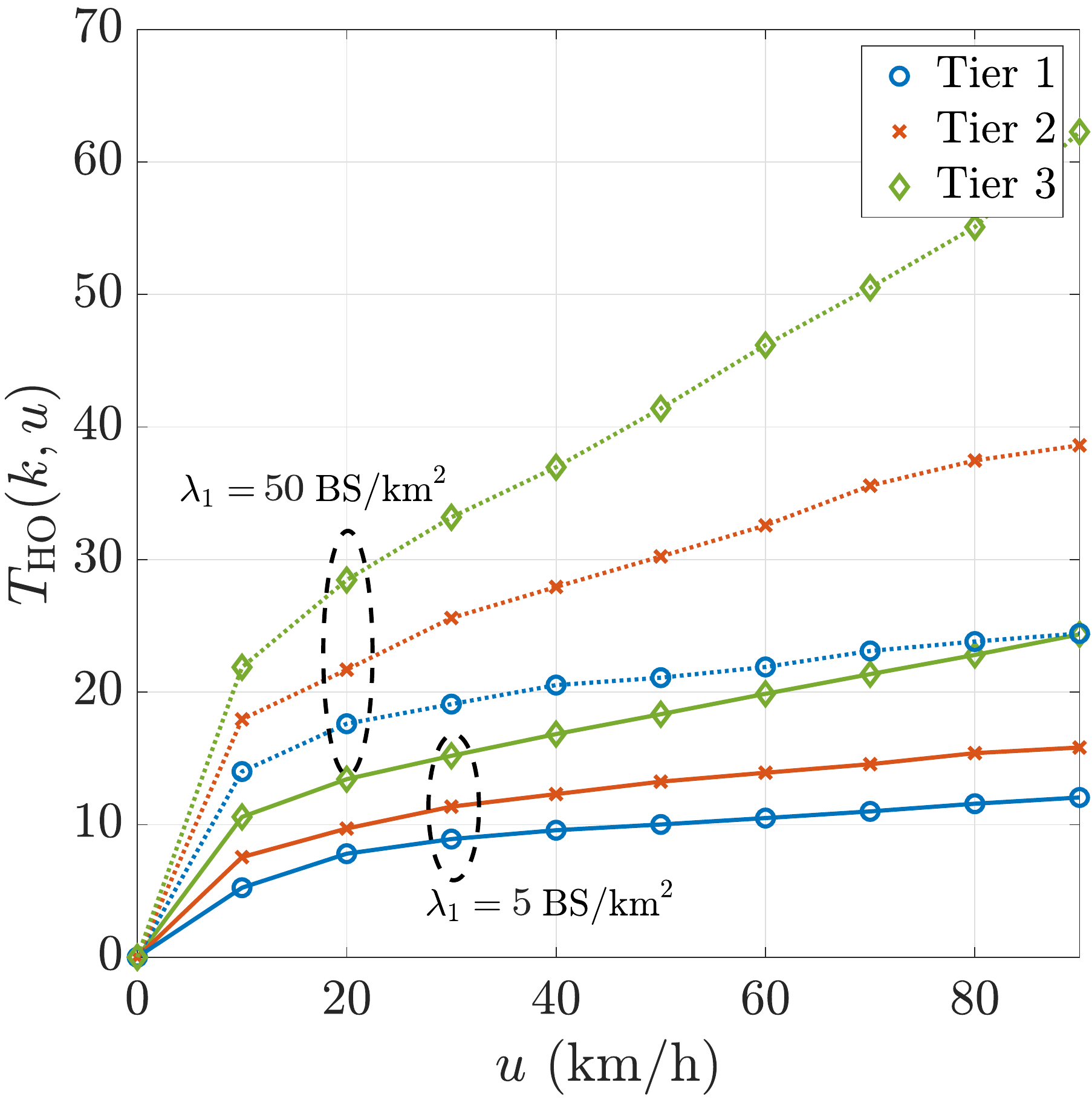}\caption{Mobility-induced time overhead versus $u$ for each access mode strategy, where $\lambda_1=\{5,50\}\ {\rm BS/km}^2$, $\lambda_2=5\lambda_1$, and $\lambda_3=10\lambda_1$}\label{ThirdFigure}\vspace{-0.5cm}
\end{figure}

\begin{figure}[!t]
	\centering
	\subfloat[]{\includegraphics[width=0.6\linewidth]{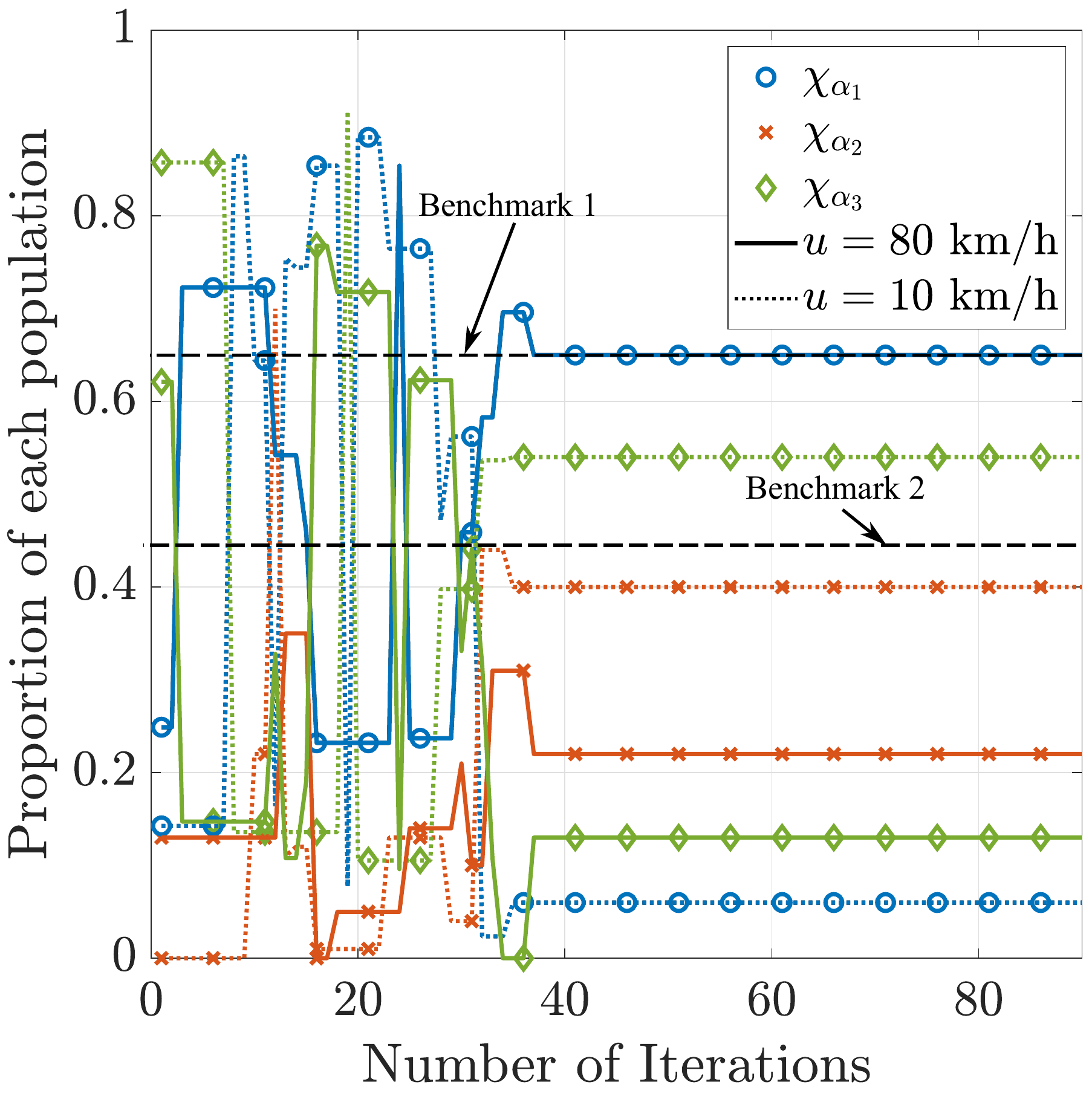}%
		\label{FigureProportionA}}
	\hfil
	\subfloat[]{\includegraphics[width=0.6\linewidth]{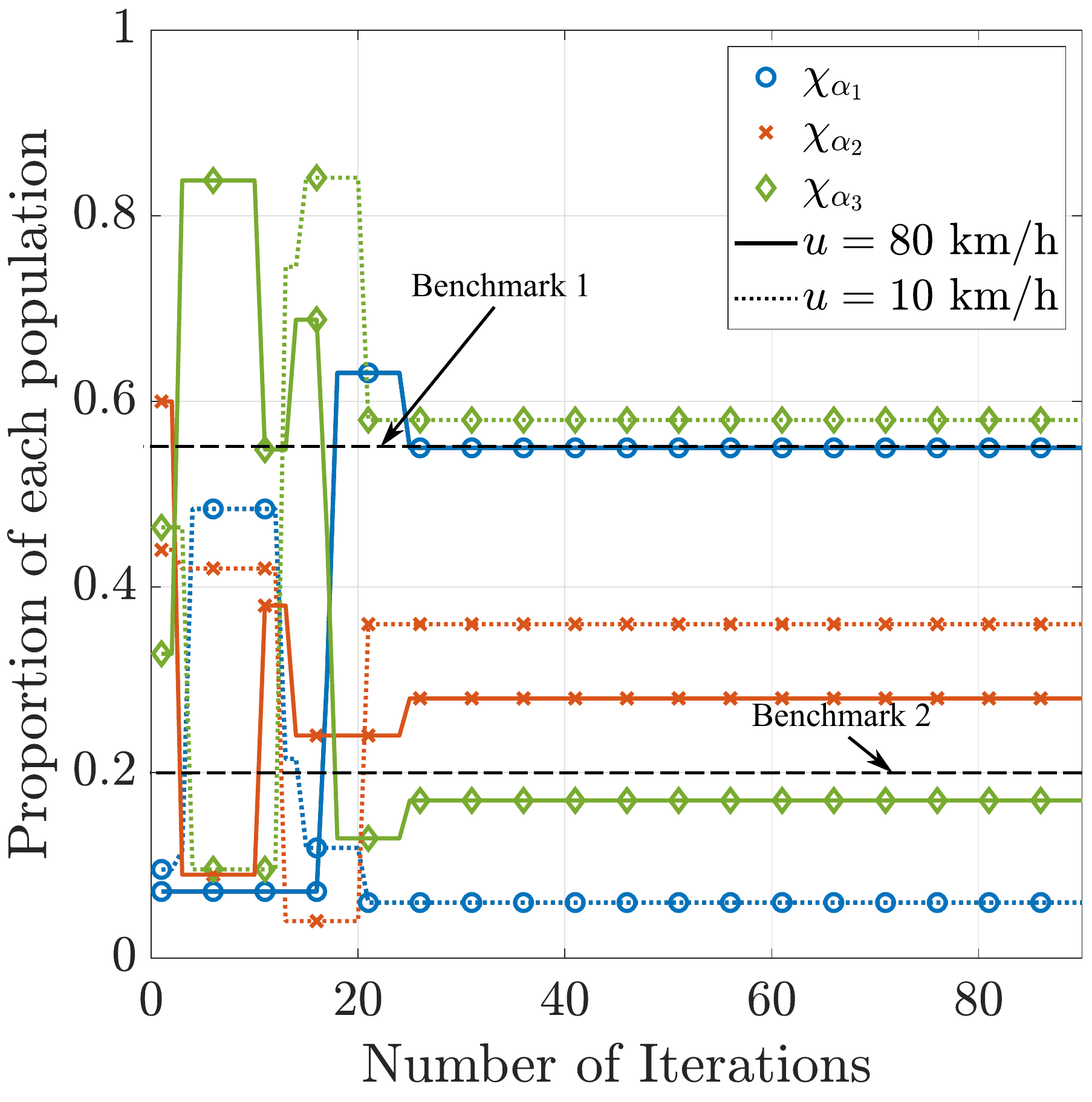}%
		\label{FigureProportionB}}
	\caption{Proportion of MUs with the proposed access mode selection algorithms, where $u=\{10,80\}$ km/h, for the proposed access mode selection policies (a) $\mathcal{G}^{(1)}$. (b) $\mathcal{G}^{(2)}$.}
	\label{fig_sim}\vspace{-0.6cm}
\end{figure}

Fig. \ref{ThirdFigure} reveals the impact of the MUs'velocity and the spatial density of BSs on the mobility-induced time overhead, $T_{\rm HO}(k,u)$, for all access mode strategies. Specifically, we plot the penalty term of each access mode strategy versus the MUs' velocity, $u$, for different spatial densities of the sub-$6$ GHz MCells i.e., $\lambda_1 = \{10^{-4},10^{-2}\}\ {\rm BS/km}^2$, where $\lambda_2=5\lambda_1$ and $\lambda_3=10\lambda_1$. An important observation from this figure is the positive effect of the BSs' density on the mobility-induced time overhead for all access mode strategies. This was expected since, the growing number of BSs leads to decreased BSs' coverage footprints, enhancing the demand of a-InterCH procedures and thus increasing the players' penalty function. Another important observation is the positive impact of MUs' velocity on the mobility-induced time overhead of each access mode strategy. This behaviour is based on the fact that, a MU with higher velocity experiences an increased number of both IntraCH and InterCH procedures, resulting in an increased penalty term (i.e., mobility-induced time overhead).

Fig. \ref{FigureProportionA} and Fig. \ref{FigureProportionB} illustrate the population evolution of each access mode strategy during the evolution process of the games $\mathcal{G}^{(1)}$ and $\mathcal{G}^{(2)}$, respectively, where $u=\{10,80\}$ km/h. Initially, we can easily observe that the proportions of both strategies eventually converge to the evolutionary equilibrium very fast, within less than $40$ iterations. An other important observation is that, both games result in same conclusions with respect to the access mode selection of a MU depending on its speed. More specifically, we can easily observe from both figures that, sub-$6$ GHz access mode is beneficial for the high-speed MUs, while mmWave access mode is more preferable for the static MUs. This was expected since, the lack of mobility of the static MUs allows their communication with the mmWave BSs, offering high data-rate performance, without handover processes. High-speed MUs, on the other hand, sacrifice their data-rate performance by connecting to sub-$6$ GHz BSs in order to reduce the number of handover processes and achieve better average payoff. Furthermore, we can observe that the game formulation $\mathcal{G}^{(2)}$ is more robust in changes of the proportion of each population compared with that of the game formulation $\mathcal{G}^{(1)}$. This was expected since, the achieved payoff functions of each population with the game formulation $\mathcal{G}^{(1)}$ are relatively equal, hence, a slight alteration in the proportion of interfering BSs lead to a significant alteration in the proportion of the MUs at each access mode strategy. Contrary, the enormous channel bandwidth allocated by the mmWave SCells results in well-separated payoff functions, and therefore, a more robust evolution of the MUs' proportion for each access mode strategy is observed. {Fig. \ref{FigureProportionA} and Fig. \ref{FigureProportionB} also compare the population evolution of high speed MUs  in Tier $1$ achieved with our proposed strategy along with that achieved with two conventional MU access mode selection strategies, namely the exhaustive search (denoted as ``Benchmark 1'') and random selection (denoted as ``Benchmark 2'') access mode selection. Note that the exhaustive search strategy always achieves optimal value, but intolerable complexity is introduced especially when the number of MUs is large, while the random selection strategy provides the worst performance.} Finally, we can observe that the game formulation $\mathcal{G}^{(2)}$ reaches the evolutionary equilibrium faster than the game formulation $\mathcal{G}^{(1)}$.

\begin{figure}
	\centering\includegraphics[width=0.6\linewidth]{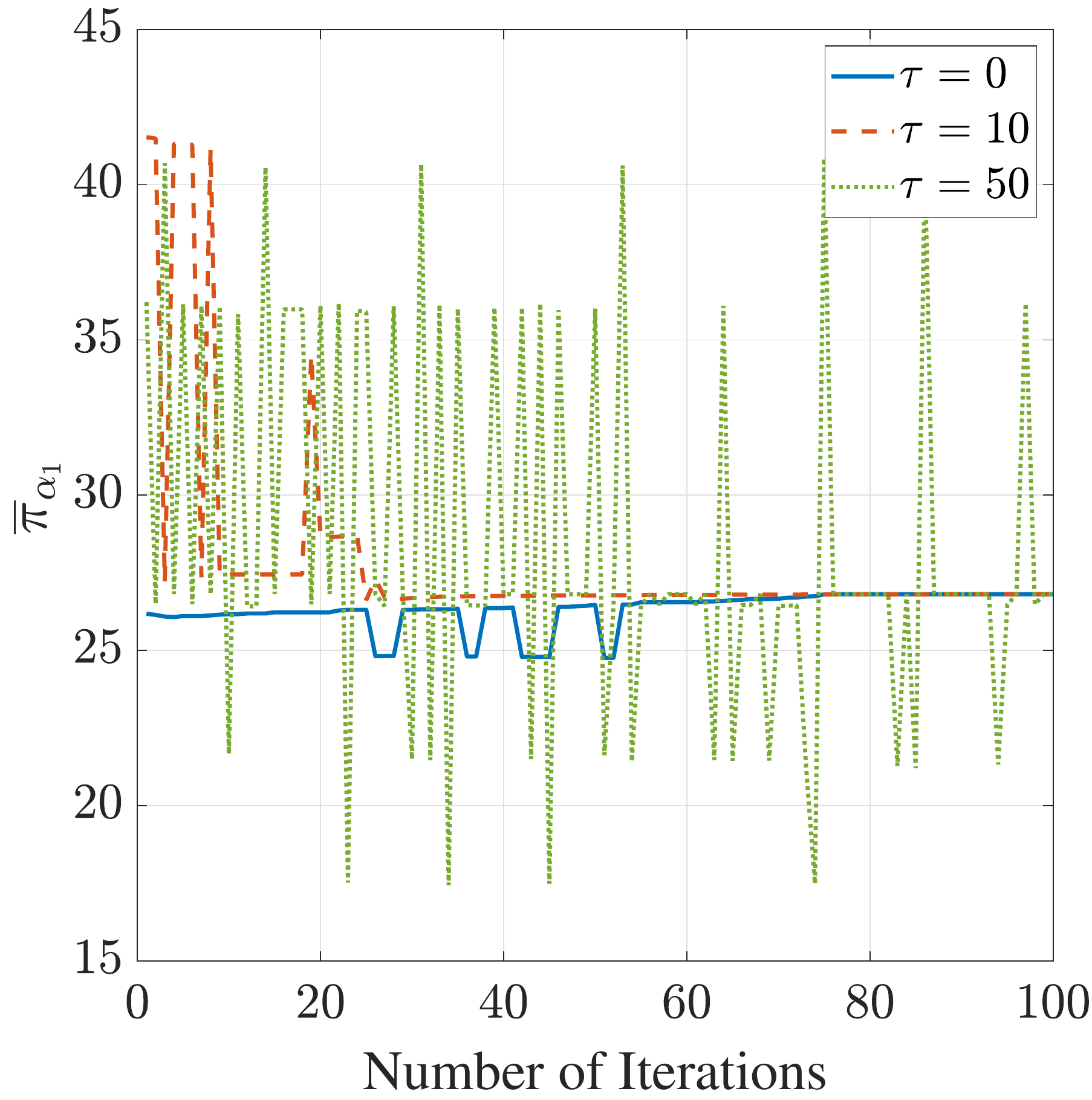}\caption{Average payoff for sub-$6$ GHz MCells versus the number of iterations for different $\tau=\{0,10,50\}$.}\label{FigureDelay}\vspace{-0.5cm}
\end{figure}

\begin{figure}
	\centering\includegraphics[width=0.6\linewidth]{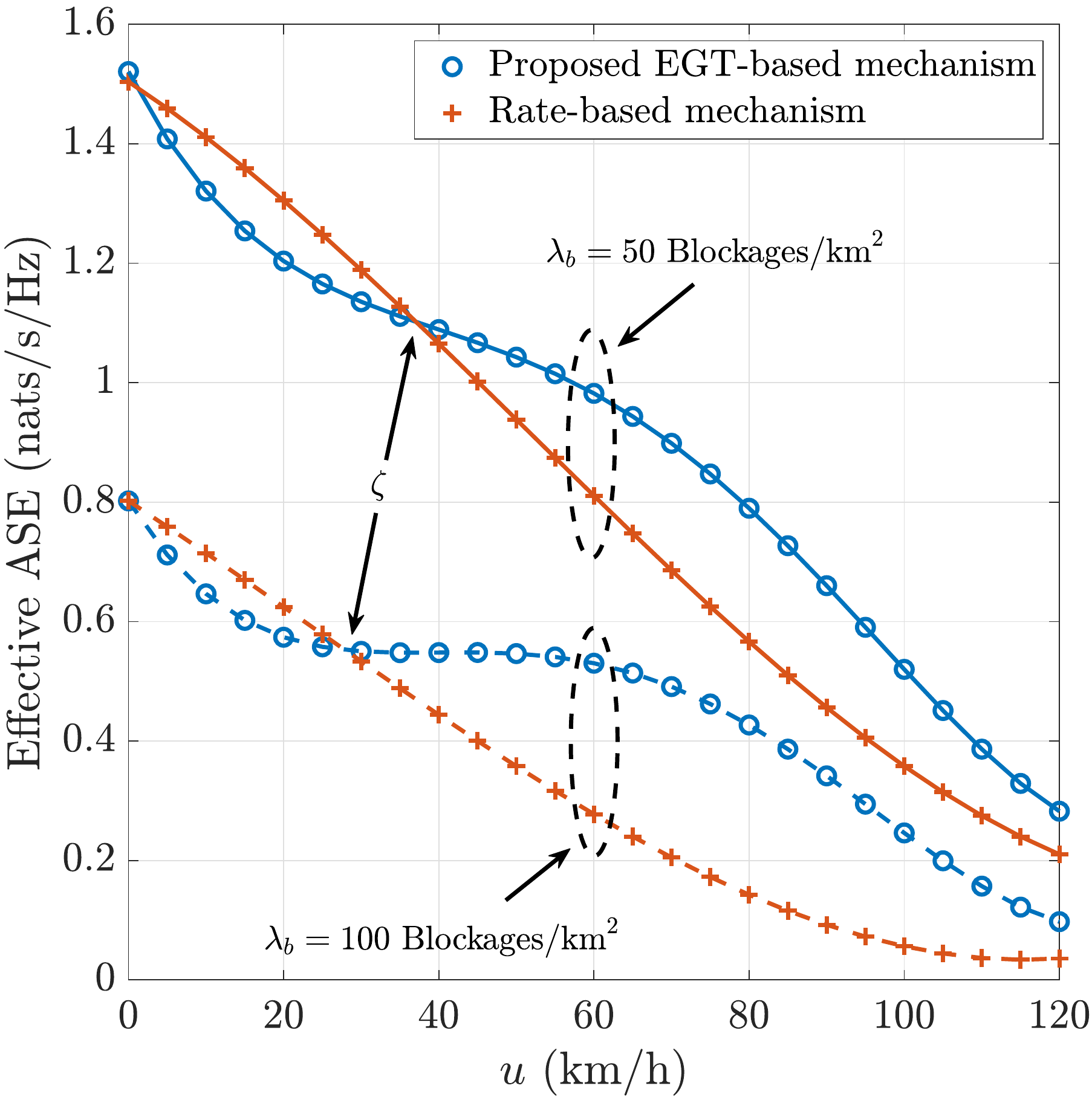}\caption{Average payoff for sub-$6$ GHz MCells versus the number of iterations for different $\tau=\{0,10,50\}$.}\label{Compare}\vspace{-0.5cm}
\end{figure}

Fig. \ref{FigureDelay} demonstrates how information exchange delay $\tau$ affects the dynamics of strategy adaptation, which is assumed to be constant throughout all the iterations. Specifically, Fig. \ref{FigureDelay} shows the average payoff af the MCell access mode strategy versus the number of iterations towards the evolutionary equilibrium. As expected, when $\tau>0$, the dynamics of the strategy adaptation exhibit large fluctuations, leading to a larger number of iterations for the dynamics to converge, and the system becomes less stable. Under a small delay, the system can still converge to the equilibrium, while for large delay, the system will diverge. It is important to mention here that, the same behavior is also observed for the rest of the access mode strategies, however for the sake of clarity, the corresponding curves are omitted from this figure.

Fig. \ref{Compare} highlights the effectiveness of our proposed technique compared to a conventional access mode selection technique for different densities of blockages $\lambda_b=\{50,100\}\ {\rm Blockages/km}^2$. More specifically, we plot the average spectral efficiency (ASE) (i.e., ${\rm ASE}_k=(1-T_{\rm HO}(k,u))^+\mathcal{R}_k$, where $k=\{1,\dots,K\}$) versus the MUs' speed for the proposed EGT-based mechanism, as well as for the conventional maximum Rate-based (denoted as ``Rate-based'') access mode selection technique. We can easily observe that, by increasing the blockage density the achieved network performance decreases for both mechanisms. This observation is based on the fact that, by increasing the blockage density, the communication of the MUs with LoS BSs becomes impossible and thus the achieved ASE significantly decreases. An important observation from this figure is that, at low MUs' velocity values, the conventional mechanism provides slightly better network performance compared to the proposed EGT-based mechanism. This was expected, since with the utilization of the conventional mechanism, a MU is served by the tier that provides the maximum rate while ignoring the negligible mobility-induced time overhead. Contrary, the EGT-based mechanism takes into consideration the time overhead incurred by the MUs' mobility, leading to the association of a MU with the tier that offers a lower average rate with a minor time overhead, and thus, in a reduced network performance. However, by increasing the MUs' velocity beyond a critical point, $\zeta$, our proposed mechanism overcomes the conventional scheme, providing a significantly enhanced network performance. As expected, our proposed technique ensures the connectivity of a MU with the network tier that provides the maximum payoff function (i.e., optimal combination of average achievable rate and mobility-induced time overhead) according to its velocity. Finally, we can easily observe that the critical velocity point, $\zeta$, reduces with the increase of the blockage density. It is straightforward that, by increasing the blockage density, the amount of handover processes as well as the required time overhead are significantly increased for both schemes, allowing our proposed EGT-based scheme to overcome the conventional scheme in lower MUs velocity values.

\section{Conclusion}\label{Conclusion}
For the purpose of addressing the new and inevitable challenges on MU association, in this paper, we have explored a novel access mode selection mechanism in the context of heterogeneous sub-$6$ GHz/mmWave cellular networks. To cope with the dynamic and complicated MUs' association process, two evolutionary games have been formulated, where the MUs, which are considered as the players, are interested in selecting a suitable access mode. By leveraging tools from SG, analytical expressions for MUs' utility and penalty/reward functions are derived, by taking into account both BS and blockage spatial randomness, as well as the MUs' mobility. The dynamics of access mode strategy adoption has been mathematically modelled by the replicator dynamics, and the evolutionary equilibrium has been considered as the stable solution for the formulated problem. Based on the average SINR experienced by a MU, the stability of the equilibrium point has been analytically proven for a two-tier heterogeneous sub-$6$ GHz/mmWave cellular network and under certain practical assumptions. For larger system configurations, the stability of the equilibrium point has been shown by numerical studies. We have also considered the impact of delayed information exchange on the convergence of the proposed algorithm. Our results show that EGT-based access mode selection algorithm is a promising solution to overcome the uncertainties imposed by the MUs' mobility in heterogeneous sub-$6$ GHz/mmWave cellular networks. Finally, we have shown that, our proposed technique offers an enhanced spectral efficiency and connectivity, when compared with the conventional access mode strategies.

\appendices
\section{Proof of Lemma \ref{LemmaLaplace}}\label{Appendix2}
The Laplace transform definition yields
	\begin{align}
		\mathcal{L}_{\mathcal{I}_k}(s) &= \mathbb{E}_{\Phi^{\rm a}_k,\mathcal{M}_k,h_{i,k}}\!\left[e^{-s\sum_{x_{i,k}\in\Phi^{\rm a}_k\backslash x_{0,k}} \mathcal{M}_k P_k h_{i,k} \|x_{i,k}\|^{-a}}\right]\nonumber\\
		&=\mathbb{E}_{\Phi^{\rm a}_k,\mathcal{M}_k}\!\left[\prod_{x_{i,k}\in\Phi^{\rm a}_k\backslash x_{0,k}}\!\frac{1}{1\!+\!s\mathcal{M}_k P_k \|x_{i,k}\|^{-a}}\right]\label{eq2}.
	\end{align}
	Since the set of active LoS BSs, $\Phi^{\rm a}_k$, is modeled as a non-homogeneous PPP with density $\lambda^{\rm a}(r)$, the above expression can be re-written as
	\begin{align}
		\!\mathcal{L}_{\mathcal{I}_k}\!(s)\!=\!\prod_{\mathcal{M}_k}\!\exp\!\left(\!-2\pi p_{\mathcal{M}_k}\!\!\int\limits_{r_0}^{\infty}\!\!\frac{s\mathcal{M}_k P_k r^{1-a}}{1\!+\!s\mathcal{M}_k P_k r^{-a}}\lambda^{\rm a}_k(r){\rm d}r\!\right),\label{eq3}
	\end{align}
	where \eqref{eq3} uses the expression for moment generating function of an exponential random variable and the probability generating functional for a PPP, and $r=\|x_{i,k}\|$. To simplify this integral, we perform the change of variable $z=(s\mathcal{M}_kP_k)^{-\frac{1}{a}} r$, obtaining the final expression.
	
\section{Proof of Lemma \ref{BlockageHandover}}\label{Appendix3}	
\begin{figure}
	\centering\includegraphics[width=0.6\linewidth]{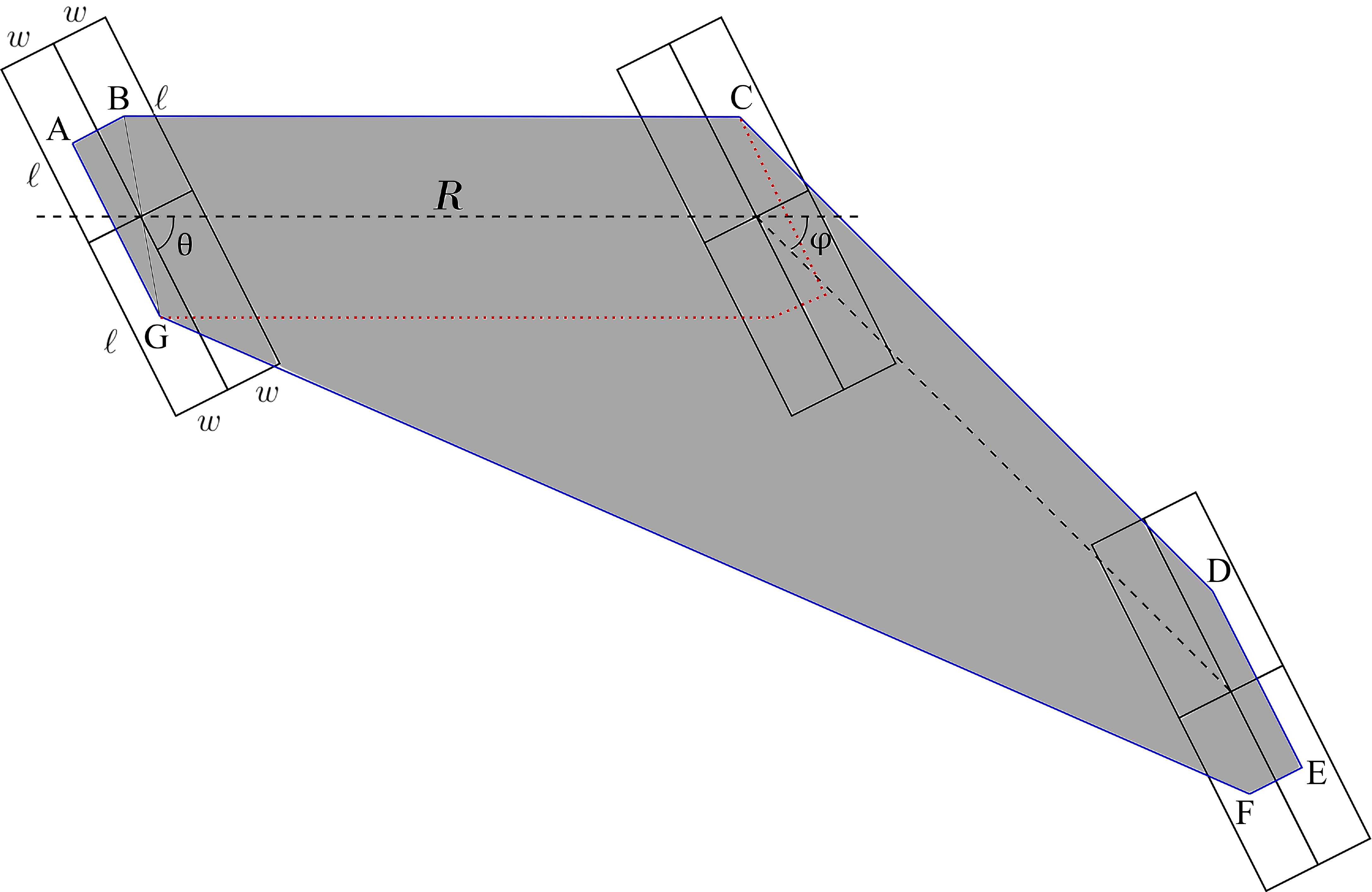}\caption{A link of distance $R$. $A$, $B$, $C$, $D$, $E$, $F$, and $G$ are the centers of the corresponding rectangles (blockages). A rectangle/blockage of of length $\ell$, width $w$, and orientation $\theta$ intersects the link if only if its center falls in the region $ABCDEFG$.
	}\label{BlockageP}\vspace{-0.5cm}
\end{figure}
Fig. \ref{Topology} illustrates the trace for the typical MU (dotted green line), which moves with velocity $u$ and is served by a BS at distance $R$. 
As shown in Fig. \ref{BlockageP}, a blockage (i.e., rectangle) intersects the link between the typical MU and its serving BS if and only if its center falls within the shaded region $S(\ell,w,\theta,\varphi,u)$. Then, conditioned on the blockage characteristics i.e., length, width, and orientation, the area of the shaded region i.e., $|S(\ell,w,\theta,\varphi,u)|$, can be calculated as
\begin{align*}
	&|S(\ell,w,\theta,\varphi,u)|\\&\qquad = \frac{Ru}{2}|\sin(\varphi)|+wu|\cos(\theta-\varphi)|+\ell u|\sin(\theta-\varphi)|\\
	&\qquad=\frac{Ru}{2}|\sin(\varphi)|+wu|\cos(\theta-\varphi)|+\ell u|\sin(\theta-\varphi)|,
\end{align*}
where $\varphi$ is the MU's movement direction and $R$ denotes the distance of the link between the initial location of the typical MU and its serving BS. Recall that the blockages' centers form a PPP with spatial density $\lambda_b$. Therefore, the average number of blockages, $\mathbb{E}[K]$, of which their centers lie within the shaded area, is given by 
\begin{align}
	\mathbb{E}[K] &= \int_L\int_W\int_\Theta \lambda_b|S(\ell,w,\theta,u)|f_L(\ell){\rm d}\ell f_W(w){\rm d}w\frac{1}{2\pi}{\rm d}\theta\\
	&=\frac{u\lambda_b}{2\pi}\left(R\zeta+\mathbb{E}[W]{\rm Si}[2\pi]\right),
\end{align}
where $\zeta$ is a constant variable that is equal to $\zeta=\gamma+\log[\pi]-{\rm Ci}[\pi]$, $\gamma$ represents the Euler–Mascheroni constant; ${\rm Si}[\cdot]$ and ${\rm Ci}[\cdot]$ denote the sine and cosine integral functions, respectively. Thus, the handover probability due to blockage effects i.e., the average rate where the link between a MU and its serving BS is blocked during the MU's movement, is given by $\delta_b=1-\mathbb{P}[K=0]$ \cite{AND}, and hence, the final expression is derived.
\section{Proof of Theorem \ref{Stability}}\label{Appendix4}
We can derive the replicator dynamics of the game $\mathcal{G}^{(1)}$ as
\begin{align}
	\dot{x}^{(1)}_1&=\varrho x^{(1)}_1\left(\pi_1^{(1)}-\chi_1^{(1)}\pi_1^{(1)}-\chi_2^{(1)}\pi_2^{(1)}\right)\nonumber\\
	&=\varrho \chi_1^{(1)}\left(\pi_1^{(1)}\left(1-\chi_1^{(1)}\right)-\chi_2^{(1)}\pi^{(1)}_2\right).
\end{align}
According to the definition i.e., $\chi_1^{(1)}+\chi_2^{(1)}=1$, the above expression can be re-written as
\begin{align}
	\dot{x}^{(1)}_1&=\varrho \chi_1^{(1)}\chi_2^{(1)}\left(\pi_1^{(1)}-\pi^{(1)}_2\right)\nonumber\\
	&=\varrho \chi_1^{(1)}(1-\chi_1^{(1)})\left(\pi_1^{(1)}-\pi^{(1)}_2\right).\label{temp}
\end{align}
In order the considered game to have a stable evolutionary equilibrium, all eigenvalues of the Jacobian of the system of equations should have a negative real part i.e., $\frac{{\rm d}f}{{\rm d}\chi_1^{(1)}}<0$, where $f$ is the right hand side of \eqref{temp}. Hence, 
\begin{align}\label{Derivative}
	\frac{{\rm d}f}{{\rm d}\chi_1^{(1)}}&=\chi_1^{(1)}(1-\chi_1^{(1)})\left(\frac{{\rm d}\pi_1^{(1)}}{{\rm d}\chi_1^{(1)}}-\frac{{\rm d}\pi_2^{(1)}}{{\rm d}\chi_1^{(1)}}\right)\nonumber\\&\qquad\qquad+(\pi_1^{(1)}-\pi_2^{(1)})(\chi_2^{(1)}-\chi_1^{(1)}).
\end{align}
Since at the equilibrium holds that $\pi_1^{(1)}=\pi_2^{(1)}$, then $(\pi_1^{(1)}-\pi_2^{(1)})(\chi_2^{(1)}-\chi_1^{(1)})=0$. Then
\begin{align}\label{DerPayoff1}
	\frac{{\rm d}\pi_1^{(1)}}{{\rm d}\chi_1^{(1)}}=-w_1\prod_{\mathcal{M}_k}\frac{p_{\mathcal{M}_k}G_{0,k}R^2}{\pi\lambda_1 (\chi_1^{(1)})^2\widetilde{r}^2 \mathcal{M}_k(R^2-\widetilde{r}^2)}
\end{align}
and 
\begin{align}\label{DerPayoff2}
	\frac{{\rm d}\pi_2^{(1)}}{{\rm d}\chi_1^{(1)}}=0
\end{align}
Therefore, based on \eqref{Derivative}, \eqref{DerPayoff1}, and \eqref{DerPayoff2}, we conclude that $\dot{x}^{(1)}_1$ at any interior equilibrium point is negative.

\begin{IEEEbiography}[{\includegraphics[width=1in,height=1.25in,clip,keepaspectratio]{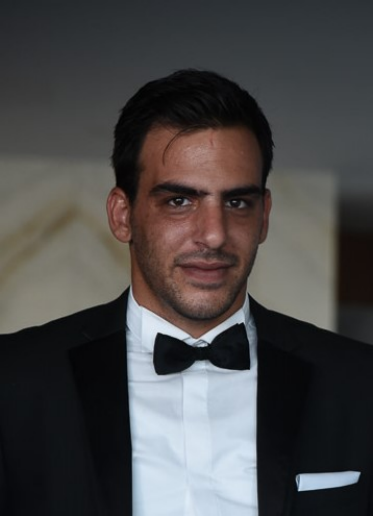}}]{Christodoulos Skouroumounis} (S'15–M'20) received the diploma in Computer Engineer from the Electrical and Computer Engineer Department of National Technical University of Athens (NTUA), Greece, in 2014, and a Ph.D. in Computer Engineer from the University of Cyprus, Cyprus in 2019. He is currently a Post-Doctoral Research Fellow with the Department of Electrical Engineering, Computer Engineering and Informatics, Cyprus University of Technology. His current research interests include full-duplex radio, next-generation communication systems and cooperative networks.
\end{IEEEbiography}

\begin{IEEEbiography}[{\includegraphics[width=1in,height=1.25in,clip,keepaspectratio]{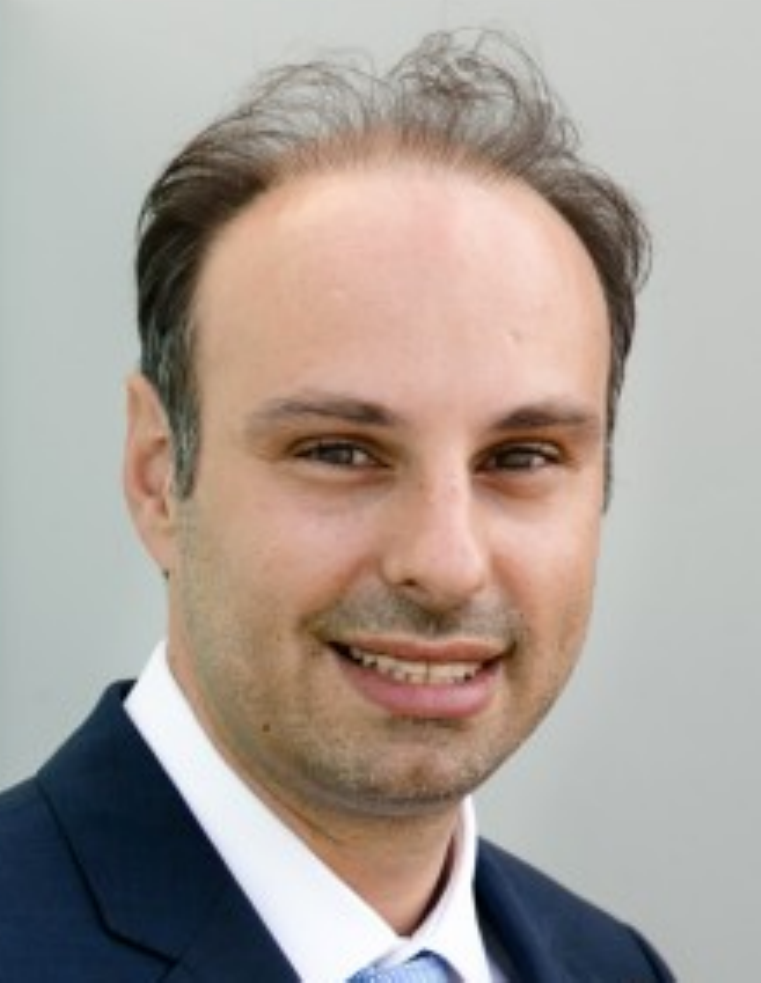}}]{Ioannis Krikidis}(S'03–M'07–SM'12–F'19) received the Diploma degree in computer engineering from the Computer Engineering and Informatics Department (CEID), University of Patras, Greece, in 2000, and the M.Sc. and Ph.D. degrees from \'{E}cole Nationale Sup\'{e}rieure des T\'{e}l\'{e}communications (ENST), Paris, France, in 2001 and 2005, respectively, all in electrical engineering. 
	
	From 2006 to 2007, he worked as a PostDoctoral Researcher with ENST, Paris. From 2007 to 2010, he was a Research Fellow of the School of Engineering and Electronics, The University of Edinburgh, Edinburgh, U.K. He also has held research positions at the Department of Electrical Engineering, University of Notre Dame; the Department of Electrical and Computer Engineering, University of Maryland; the Interdisciplinary Centre for Security, Reliability and Trust, University of Luxembourg; and the Department of Electrical and Electronic Engineering, Niigata University, Japan. He is currently an Associate Professor at the Department of Electrical and Computer Engineering, University of Cyprus, Nicosia, Cyprus. His current research interests include wireless communications, cooperative networks, 5G communication systems, wireless powered communications, and secrecy communications. He serves as an Associate Editor for IEEE TRANSACTIONS ON WIRELESS COMMUNICATIONS, IEEE TRANSACTIONS ON GREEN COMMUNICATIONS AND NETWORKING, and IEEE WIRELESS COMMUNICATIONS LETTERS. He was a recipient of the Young Researcher Award from the Research Promotion Foundation, Cyprus, in 2013, and a recipient of IEEE ComSoc Best Young Professional Award in Academia, 2016. He has been recognized by the Web of Science as a Highly Cited Researcher for 2017, 2018, and 2019. He has received the prestigious ERC Consolidator Grant.
\end{IEEEbiography}

\begin{thebibliography}{10}
	
\bibitem{AKY} I. F. Akyildiz, A. Kak, and S. Nie, ``6G and beyond: The future of wireless communications systems," \emph{IEEE Access,} vol. 8, pp. 133995--134030, 2020.

\bibitem{ZHAN} Z. Zhang, Y. Xiao, Z. Ma, M. Xiao, Z. Ding, X. Lei, G. Karagiannidis, and P. Fan, ``6G wireless networks: Vision, requirements, architecture, and key technologies,'' \emph{IEEE Veh. Technology Mag.,} vol. 14, no. 3, pp. 28--41, Sept. 2019.

\bibitem{XU} X. Xu, Z. Sun, X. Dai, T. Svensson, and X. Tao, ``Modeling and analyzing the cross-tier handover in heterogeneous networks,'' \emph{IEEE Trans. Wireless Commun.,} vol. 16, no. 12, pp. 7859--7869, Dec. 2017.

\bibitem{WAN} M. Xiao, S. Mumtaz, Y. Huang, L. Dai, Y. Li, M. Matthaiou, G. K. Karagiannidis, E. Bjornson, K. Yang, Chin-Lin I, and A. Ghosh, ``Millimeter wave communications for future mobile networks,'' \emph{IEEE J. Sel. Areas Commun.,} vol. 35, no. 9, pp. 1909--1935, Sept. 2017.

\bibitem{AND} J. G. Andrews, T. Bai, M. N. Kulkarni, A. Alkhateeb, A. K. Gupta, and R. W. Heath, ``Modeling and analyzing millimeter wave cellular systems,'' \emph{IEEE Trans. Commun.,} vol. 65, no. 1, pp. 403--430, Jan. 2017.

\bibitem{KAL} S. S. Kalamkar, F. Baccelli, F. M. Abinader, A. S. Marcano Fani, and L. G. Uzeda Garcia, ``Beam management in 5G: A stochastic geometry analysis,'' accepted in \emph{IEEE Trans. Wireless Commun.,}.

\bibitem{OZK} M. F. O\"{z}ko\c{c}, A. Koutsaftis, R. Kumar, P. Liu, and S. S. Panwar, ``The impact of multi-connectivity and handover constraints on millimeter wave and terahertz cellular networks,'' \emph{IEEE J. Select. Areas Commun.,} .

\bibitem{LIU2} Y. Liu, X. Fang, M. Xiao and S. Mumtaz, ``Decentralized beam pair selection in multi-beam millimeter-wave networks,'' \emph{IEEE Trans. Commun.,} vol. 66, no. 6, pp. 2722--2737, June 2018.

\bibitem{ELS} H. Elshaer, M. N. Kulkarni, F. Boccardi, J. G. Andrews, and M. Dohler, ``Downlink and uplink cell association with traditional macrocells and millimeter wave small cells,'' \emph{IEEE Trans. Wireless Commun.,} vol. 15, no. 9, pp. 6244--6258, Sept. 2016.

\bibitem{GHA} G. Ghatak, A. De Domenico, and M. Coupechoux, ``Coverage analysis and load balancing in HetNets with millimeter wave multi-RAT small cells,'' \emph{IEEE Trans. Wireless Commun.,} vol. 17, no. 5, pp. 3154--3169, 2018.

\bibitem{SHI2} M. Shi, K. Yang, C. Xing, and R. Fan, ``Decoupled heterogeneous networks with millimeter wave small cells,'' \emph{IEEE Trans. Wireless Commun.,} vol. 17, no. 9, pp. 5871--5884, Sept. 2018.

\bibitem{LIU} R. Liu, Q. Chen, G. Yu, and G. Y. Li, ``Joint user association and resource allocation for multi-band millimeter-wave heterogeneous networks,'' \emph{IEEE Trans. Commun.,} vol. 67, no. 12, pp. 8502--8516, Dec. 2019.

\bibitem{ZHA2} H. Zhang, S. Huang, C. Jiang, K. Long, V. C. M. Leung, and H. V. Poor, ``Energy efficient user association and power allocation in millimeter-wave-based ultra dense networks with energy harvesting base stations,'' \emph{IEEE J. Selected Areas Commun.,} vol. 35, no. 9, pp. 1936--1947, Sept. 2017.

\bibitem{SKO} C. Skouroumounis, C. Psomas, and I. Krikidis, ``A hybrid cooperation scheme for sub-6 GHz/mmWave cellular networks,'' \emph{IEEE Commun. Lett.,} vol. 24, no. 7, pp. 1539--1543, July 2020.

\bibitem{POL} M. Polese, M. Giordani, M. Mezzavilla, S. Rangan, and M. Zorzi, ``Improved handover through dual connectivity in 5G mmWave mobile networks,'' \emph{IEEE J. Select. Areas Commun.,} vol. 35, no. 9, pp. 2069--2084, Sept. 2017.

\bibitem{KIB} M. G. Kibria, K. Nguyen, G. P. Villardi, W. Liao, K. Ishizu, and F. Kojima, ``A stochastic geometry analysis of multiconnectivity in heterogeneous wireless networks,'' \emph{IEEE Trans. Veh. Tech.,} vol. 67, no. 10, pp. 9734--9746, Oct. 2018.

\bibitem{CHO} S. Choi, J. Choi, and S. Bahk, ``Mobility-aware analysis of millimeter wave communication systems with blockages,'' \emph{IEEE Trans. Veh. Tech.,} vol. 69, no. 6, pp. 5901--5912, Jun. 2020.

\bibitem{MON} V. F. Monteiro, M. Ericson, and F. R. P. Cavalcanti, ``Fast-RAT scheduling in a 5G multi-RAT scenario,'' \emph{IEEE Commun. Mag.,} vol. 55, no. 6, pp. 79--85, June 2017.

\bibitem{ZHA} J. Zhao, S. Zhao, H. Qu, G. Ren, and Y. Shi, ``Analysis and optimization of probabilistic caching in micro/millimeter wave hybrid networks with dual connectivity,'' \emph{ IEEE Access,} vol. 6, pp. 72372--72380, 2018.

\bibitem{YAN1} L. Yan, H. Ding, L. Zhang, J. Liu, X. Fang, Y. Fang, M. Xiao, and X. Huang, ``Machine learning-based handovers for sub-6 GHz and mmWave integrated vehicular networks,'' \emph{IEEE Trans. Wireless Commun.,} vol. 18, no. 10, pp. 4873--4885, Oct. 2019.

\bibitem{NIY} D. Niyato and E. Hossain, ``Dynamics of network selection in heterogeneous wireless networks: An evolutionary game approach,'' \emph{IEEE Trans. Veh. Technology,} vol. 58, no. 4, pp. 2008--2017, May 2009.

\bibitem{YAN2} S. Yan, M. Peng, M. A. Abana, and W. Wang, ``An evolutionary game for user access mode selection in fog radio access networks,'' \emph{IEEE Access,} vol. 5, pp. 2200--2210, 2017.

\bibitem{SEM} P. Semasinghe, E. Hossain, and K. Zhu, ``An evolutionary game for distributed resource allocation in self-organizing small cells,'' \emph{IEEE Trans. Mobile Computing,} vol. 14, no. 2, pp. 274--287, 1 Feb. 2015.

\bibitem{HAEb} M. Haenggi, {\it Stochastic geometry for wireless networks}, in \emph{Cambridge}, U.K.: Cambridge Univ. Press, 2012.

\bibitem{BAN} M. Banagar and H. S. Dhillon, ``3GPP-inspired stochastic geometry-based mobility model for a drone cellular network,'' in \emph{Proc. IEEE Global Commun. Conf.,} Dec. 2019.

\bibitem{GRA} I. S. Gradshteyn and I. M. Ryzhik, {\it Table of integrals, series, and products}, in \emph{Elsevier}, Academic Press, 2007.

\bibitem{ALZ} H. Alzer, ``On some inequalities for the incomplete Gamma function,'' \emph{Mathematics of Computation,} pp. 771--778, 1997.

\bibitem{JAI} A. Jain, E. Lopez-Aguilera, and I. Demirkol, ``Improved handover signaling for 5G networks,'' in \emph{Proc. IEEE Annual Int. Symp. Personal, Indoor and Mobile Radio Commun.,} 2018, pp. 164--170.

\bibitem{ARS} R. Arshad, H. ElSawy, S. Sorour, T. Y. Al-Naffouri, and M. Alouini, ``Handover management in dense cellular networks: A stochastic geometry approach,'' in \emph{Proc. IEEE Int. Conf. Commun.,} 2016, pp. 1-7.

\bibitem{BAI} T. Bai, R. Vaze, and R. W. Heath, ``Analysis of blockage effects on urban cellular networks,'' \emph{IEEE Trans. Wireless Commun.,} vol. 13, no. 9, pp. 5070--5083, Sept. 2014.
\end{thebibliography}
\end{document}